%% file: root.tex
\documentclass[journal]{IEEEtran}

\usepackage{amsmath, amsthm, amssymb, dsfont, accents, xcolor, mathdots}
\usepackage{booktabs,hyperref}

\usepackage{todonotes}

\usepackage{tikz,pgfplots}
\usetikzlibrary{calc, arrows.meta, external, backgrounds, shapes}
\pgfplotsset{compat=newest}

\newtheorem{definition}{Definition}
\newtheorem{proposition}{Proposition}
\newtheorem{problem}{Problem}
\newtheorem{theorem}{Theorem}
\newtheorem{corollary}{Corollary}
\newtheorem{lemma}{Lemma}
\newtheorem{remark}{Remark}
\newtheorem{assumption}{Assumption}
\newtheorem{example}{Example}

\input{defs}

\usepackage[linesnumbered,ruled]{algorithm2e}

\begin{document}

\title{Control Synthesis for Permutation-Symmetric High-Dimensional Systems With Counting Constraints}

\author{Petter~Nilsson,~\IEEEmembership{Member,~IEEE,}
        and Necmiye~Ozay,~\IEEEmembership{Member,~IEEE.}
\thanks{P. Nilsson is with the Department of Mechanical and Civil Engineering at California Institute of Technology, Pasadena, CA USA. Email: pettni@caltech.edu. N. Ozay is with the Electrical Engineering and Computer Science Department at the University of Michigan, Ann Arbor, MI USA. Email: necmiye@umich.edu}
}

\maketitle

\begin{abstract}
General purpose correct-by-construction synthesis methods are limited
to systems with low dimensionality or simple specifications. In this 
work we consider highly symmetrical \emph{counting problems} and exploit 
the symmetry to synthesize provably correct controllers   
for systems with tens of thousands of states. The key ingredients of the solution
are an aggregate abstraction procedure for mildly heterogeneous systems, and a formulation of counting constraints as linear inequalities.
\end{abstract}


\section{Introduction}
\label{sec:introduction}

\IEEEPARstart{A}{utomated} controller synthesis for systems subject to an a priori given specification is an attractive means of controller design; such \emph{correct-by-construction} synthesis methods have attracted considerable interest in the past years \cite{Belta:2017wo}. However, all methods that are capable of solving general problems face fundamental limitations in terms of the system dimensionality and specification complexity that can be handled. A large body of work resorts to abstraction procedures \cite{tabuada2009verification}, where a finite representation of the system is obtained for which synthesis techniques for finite-state systems can be applied, but naive grid-based abstraction methods quickly run into the curse of dimensionality since, in general, the number of discrete states grows exponentially with system dimension. Other methods operate directly on a continuous state space but also ``blow up'' when the representation and treatment of high-dimensional sets become overwhelmingly expensive.

Ways to overcome these limitations in the abstraction step have been proposed. A large problem can sometimes be decomposed into several lower-dimensional problems \cite{Conte:2016in,sno_cdc2016}, where each subproblem is easier to solve. Such methods typically require the decomposition to be given, and that the coupling between the subproblems is relatively weak. Similarly, compositional abstraction procedures have been suggested as a method to construct abstractions with a size that does not scale exponentially with dimensionality \cite{Pola2016,Pola2017}; to capture the dynamics more accurately the abstract subsystems can be made partly overlapping \cite{Meyer2017}. These compositional approaches account for dynamic coupling between subsystems; as a contrast, we consider in this paper subsystems that are dynamically independent but coupled through a common specification. Other work has exploited monotonicity \cite{coogan2015efficient,Kim:2016:DSA:2883817.2883833} to alleviate the curse of dimensionality when constructing an abstraction, and it has been found that under certain stability conditions discretization can be circumvented completely by constructing an abstraction whose states are a finite-horizon input history \cite{Tarraf2014,Zamani2015}.  

In this work we explore a different approach to tackle high-dimensionality: exploitation of symmetries. The system we consider is the aggregate system consisting of $N$ almost-homogeneous switched subsystems, where $N$ can be very large. Due to homogeneity the overall system exhibits symmetries in the dynamics which can be leveraged to enable synthesis for a special type of symmetric constraints called \emph{counting constraints}. Like many existing methods in the literature, our solution approach is abstraction-based, but with the crucial difference that we construct only one abstraction that is used for all subsystems, thus roughly reducing the number of states in the abstraction from $\mathcal O(1/\eta^{Nn_x})$ in the na\"ive approach to $\mathcal O(1/\eta^{n_x})$, where $n_x$ is the dimension of a single subsystem and $\eta$ the precision. Our main contributions are not abstraction construction itself, but showing how a single abstraction can be utilized for control synthesis purposes.\footnote{We consider here the problem of synthesizing open-loop infinite-horizon trajectories that satisfy counting constraints for a given initial condition, rather than the more general problem of synthesizing feedback controllers that enforce specifications for sets of initial conditions.} We are able to account for mild heterogeneity among the subsystems by assuming certain stability conditions. Utilizing a single abstraction for multiple subsystems allows us to synthesize controlled trajectories for problems of very high dimension that exhibit these symmetries---we demonstrate the method on a $10,000$-dimensional switched system with $2^{10,000}$ modes. If there are several distinct classes of subsystems, one abstraction can be constructed for each class; and complexity increases linearly with the number of classes.

This work is motivated by the problem of scheduling of thermostatically controlled loads. Examples of TCLs include air conditioners, water heaters, refrigerators, etc., that operate around a temperature set point by switching between being $\on$ and $\off$. TCL owners are typically indifferent to small temperature perturbations and accept temperatures in a range around their desired set point; this range is called the \emph{dead band}. The idea behind TCL scheduling is that an electric utility company can leverage the implied flexibility—--which becomes meaningful for large collections of TCLs--—to shape aggregate demand on the grid. Previous work has resulted in control algorithms based on broadcasting a universal set-point temperature \cite{Ghaffari:2015fo}, Markov chain bin models \cite{koch2011modeling,esmaeil2015aggregation}, and priority stacks \cite{Hao:2015fp}, with the objective to track a power signal. However, they do not take into account any hard constraints on aggregate power consumption (i.e., the number of TCLs in mode $\on$), which is crucial to do to avoid overloading the grid or underutilizing the power generated by renewable resources.

The fundamental problem in TCL scheduling is to simultaneously meet local safety constraints (i.e., maintain each TCL in its deadband), and global aggregate constraints. These constraints are special instances of \emph{counting constraints}, which is the type of constraints we consider in this work. The TCL scheduling example will be used throughout the paper to motivate our discussion:
\begin{example}
	Let $\{ \xc_n \}_{n \in [N]}$ be the states of a family of $N$ switched systems of the following form:
	\begin{equation}
		\frac{\mathrm{d}}{\mathrm{d}t} \xc_n(t) = f_{\u_n(t)} \left( \xc_n(t), \dc_n(t) \right), \quad \u_n(t) \in \{ \on, \off \},
	\end{equation}
	where $f_\on$ and $f_\off$ are functions $\mathbb{R} \times \mathcal D \rightarrow \mathbb{R}$ for some bounded disturbance set $\mathcal D$. For a desired dead band $[\munderbar a, \moverbar a]$, and bounds $[\munderbar K, \moverbar K]$ on the aggregate number of TCLs that are in mode $\on$, synthesize switching protocols that enforce
	\begin{align}
		& \label{eq:tcl_1} \text{$\xc_n(t) \in [\munderbar a, \moverbar a]$ for all $n \in [N]$ and all $t \in \mathbb{R}^+$}, \\
		& \label{eq:tcl_2} \begin{aligned}
			& \text{For all $t \in \mathbb{R}^+$, \textbf{at least} $\munderbar K$ and \textbf{at most} $\moverbar K$} \\ 
			& \text{ of the TCLs are in mode $\on$.}
		\end{aligned}
	\end{align}
\end{example}

Our approach can be seen as a marriage between ``bin'' abstraction previously used in the TCL literature, and schedule search via integer programming previously considered e.g. for intersection controllers \cite{colombo2015least}---generalized to arbitrary incrementally stable systems and an infinite horizon.

A preliminary version of this work appeared in \cite{no_hscc2016} where the \emph{mode}-counting problem was introduced for homogeneous collections of systems. In this paper we generalize the counting problem to encompass both mode- and state-counting constraints; make the construction of the abstractions robust to model deviations---thus allowing mild heterogeneity; improve the rounding scheme that enables search for solutions via (non-integer) linear programming; and sharpen analytical results. A different extension of counting problems that allows for a richer class of specifications---temporal logic formulas defined over counting propositions---was reported in \cite{sno_2016}.

\vspace{-0.4cm}
\subsection{Overview and Paper Structure}


We give a high-level overview of the proposed solution approach to the problem of synthesizing controllers for large collections of systems with counting constraints, namely the counting problem as formally defined in Section \ref{sec:mode_counting_problem}. By exploiting the fact that the collection is almost homogeneous, we first construct a single finite abstraction that approximates all the systems in the collection (Section \ref{sec:cont_to_discrete}). By means of this abstraction, we transform a continuous-state counting problem into a discrete one. We then treat this finite abstraction as a histogram where each abstract state corresponds to a bin of the histogram and counts the number of systems whose continuous-state is in some equivalence class associated with that abstract state (Section \ref{sec:solving_discrete}). As the continuous states evolve in time according to the chosen control inputs, the counts on this histogram also evolve according to certain aggregate dynamics. In particular, the evolution of the histogram can be represented as a constrained linear system over an integer lattice, and the counting constraints become linear constraints on the states and input of the aggregate system (Section \ref{sub:aggregate_dynamics_as_a_linear_system}). This representation allows us to reduce the discrete counting problem to the feasibility of an integer linear program (Section \ref{sub:solution_to_the_graph_mode_counting_problem}). Through this integer linear program, we search for inputs to steer the systems to periodic trajectories, captured by the cycles of the abstraction graph, so that the satisfaction of the counting constraints can be guaranteed indefinitely. We then analyze this integer linear program and a relaxation of the integrality constraints, and their relation to the (in)-feasibility of the original problem (Section \ref{sec:analysis}). An extension of the theory to the multi-class setting is briefly presented in Section \ref{sec:multiclass}. To highlight the applicability of the method, one numerical example and one application-motivated example are provided in Section \ref{sec:examples} before the paper is summarized in Section \ref{sec:conclusion}. In an effort to improve readability, technical proofs are deferred to the appendix.

\vspace{-0.3cm}
\section{Notation and Preliminaries}
\label{sec:preliminaries}

We introduce some notation that is used throughout the paper. The set of real numbers is denoted $\mathbb{R}$, the set of positive reals $\mathbb{R}_+$, and the set of non-negative integers $\mathbb{N}$. To express a finite set of positive integers, we write $[N] = \{ 0, \ldots, N-1 \}$. The indicator function of a set $X$ is denoted $\mathds{1}_{X}(x)$ and is equal to 1 if $x \in X$ and to $0$ otherwise. The identity function on a set $A$ is written as $\text{Id}_A$. For two sets $X$ and $Y$, we write the Minkowski sum as $X \oplus Y = \{ x + y : x \in X, y \in Y \}$, and the Minkowski difference $X \ominus Y = \{ x : \{x\} \oplus Y \subset X \}$. Set complement is denoted $X^C$.

To denote the floor and ceiling of a number we write $\left \lfloor \cdot \right \rfloor$ and $\left \lceil \cdot \right \rceil$. We use the same notation for vectors, where the operations are performed component-wise. We write the infinity norm as $\| \cdot \|_\infty$ and the $1$-norm as $\| \cdot \|_1$. The $\epsilon$-ball in $p-$norm centered at the point $x$ is denoted as $\mathcal B_p(x, \epsilon) = \{ y : \| y-x \|_p \leq \epsilon \} $. The vector of all $1$'s is written as $\mathbf{1}$. For a function $f : X \rightarrow \mathbb{R}$ with a finite domain we abuse notation and write $\| f \|_1 = \sum_{x \in X} |f(x)|$ for the ``1-norm'' of the finite image set. The function that is constantly equal to $0$ is written $\mathbf{0}$.

Given an ODE $\frac{\mathrm{d}}{\mathrm{d}t} \xc(t) = f_\du(\xc(t),\dc(t))$, where $\dc(t)$ is an uncontrolled input, the corresponding flow operator is denoted $\phi_\du(t,x,\dc)$ and has the properties that $\phi_\du(0,x,\dc) = x$, $\frac{\mathrm{d}}{\mathrm{d}t} \phi_\du(t,x,\dc) = f_\du(\phi_\du(t,x,\dc), \dc(t))$. A $\mathcal {KL}$-function $\beta : \mathbb{R}_+ \times \mathbb{R}_+ \rightarrow \mathbb{R}_+$ is characterized by being strictly increasing from 0 in the first argument and decreasingly converging to 0 in the second argument.

\vspace{-0.4cm}
\subsection{Transition Systems and Simulations}

We employ the following definition for a \emph{transition system}, which captures systems with both continuous and discrete state spaces.
\begin{definition}
\label{def:trans_syst}
	A \textbf{transition system} (TS) is a tuple $\Sigma = (\mathcal Q, \mathcal U, \longrightarrow, Y)$, where $\mathcal Q$ is a set of states, $\mathcal U$ a set of actions (inputs), $\longrightarrow \subset \mathcal Q \times \mathcal U \times \mathcal Q$ a transition relation, and $Y: \mathcal Q \longrightarrow \mathcal Y$ an output function. We say that $\Sigma$ is a \textbf{deterministic finite transition system} (DFTS) if i) transitions are deterministic, i.e., if $(\ds,\du,\ds') \in \mathcal \longrightarrow$ and $(\ds,\du,\ds'') \in \mathcal \longrightarrow$, then $\ds' = \ds''$, and ii) $\mathcal Q$ is finite.
\end{definition}
For simplicity, existence of a transition $(\ds, \du, \ds') \in \longrightarrow$ is written $\ds \overset{\du}{\longrightarrow} \ds'$. By a trajectory of a transition system, we mean a sequence $\xc(0)\xc(1)\xc(2)\ldots$ of states in $\mathcal Q$ with the property that $\xc(s) \overset{\u(s)}{\longrightarrow} \xc(s+1)$ for some $\u(s) \in \mathcal U$, for all $s \in \mathbb{N}$.

In this paper we restrict attention to finite input sets $\mathcal U$; such transition systems are obtained when a continuous-time switched system is time-discretized, and in this case the inputs are the modes of the switched system. A \emph{switching protocol} is then a function $\pi$ that generates control inputs from information about the current state\footnote{A switching protocol may also have internal memory states.}. A trajectory generated by a switching protocol $\pi$ is a trajectory where $\u(s)$ is generated by $\pi$.

For two systems with the same input space $\mathcal U$ and normed output space $\mathcal Y$, we adopt the following notion of system bisimilarity \cite{Girard:2007we}.
\begin{definition}
\label{def:bisimilar}
	Two transition systems $(\mathcal Q_1, \mathcal U, \longrightarrow_1, Y_1)$ and $(\mathcal Q_2, \mathcal U, \longrightarrow_2, Y_2)$ are \textbf{$\epsilon$-approximately bisimilar} if there exists a relation $R \subset \mathcal Q_1 \times \mathcal Q_2$ such that the sets $R(\ds_1) = \{ \ds_2 : (\ds_1, \ds_2) \in R \}$ and $R^{-1}(\ds_2) = \{ \ds_1 : (\ds_1, \ds_2) \in R \}$ are non-empty for all $\ds_1, \ds_2$, and such that for all $(\ds_1, \ds_2) \in R$,
	\begin{enumerate}
		\item $\|  Y_1(\ds_1) - Y_2(\ds_2)\|_\infty \leq \epsilon$,
		\item if $\ds_1 \overset{\du}{\longrightarrow_1} \ds_1'$, there exists $\ds_2 \overset{\du}{\longrightarrow_2} \ds_2'$ s.t $(\ds_1' ,\ds_2') \in R$,
		\item if $\ds_2 \overset{\du}{\longrightarrow_2} \ds_2'$, there exists $\ds_1 \overset{\du}{\longrightarrow_1} \ds_1'$ s.t $(\ds_1', \ds_2') \in R$.	
	\end{enumerate}
\end{definition}

\subsection{Graphs}

The following standard notions are used for a directed graph $G = (\mathcal Q,E)$ with node set $\mathcal Q$ and edge set $E \subset \mathcal Q \times \mathcal Q$. A \emph{path} in $G$ is a list of edges $(\ds_0, \ds_1) (\ds_1, \ds_2) \ldots (\ds_{J-1}, \ds_J)$. The distance between two nodes $\ds_0$ and $\ds_J$ is the length of the shortest path connecting the two. The graph diameter $\diam(G)$ is the longest distance between two nodes in the graph. If the first and last nodes in a path are equal, i.e. $\ds_J = \ds_0$, the path is a \emph{cycle}. A cycle is \emph{simple} if it visits every node at most one time. For a subset of nodes $D \subset \mathcal Q$ (or the corresponding subgraph, which we use interchangeably), it is said to be \emph{strongly connected} if for each node pair $(\ds,\ds') \in D$, there exists a path from $\ds$ to $\ds'$. Any directed graph can be decomposed into strongly connected components. The \emph{period} of a subgraph $D$ is the greatest common divisor of the lengths of all cycles in $D$. A subgraph $D$ is called \emph{aperiodic} if it has period one.

\section{The Counting Problem}
\label{sec:mode_counting_problem}

We first define the concept of a counting constraint.

\begin{definition}
	A \textbf{counting constraint} $(\cS, \cN)$ for a collection of $N$ subsystems numbered from $0$ to $N-1$, with the same state space $\mathcal Q$ and input space $\mathcal U$, is a set $\cS = \cS^{\mathcal Q} \times \cS^{\mathcal U}$ with $\cS^{\mathcal Q} \subset \mathcal Q$ and $\cS^{\mathcal U} \subset \mathcal U$, and a bound $\cN$. We say that the counting constraint is \textbf{satisfied} by state-input pairs $(\ds_n, \du_n)$, $n \in [N]$, if the number of pairs that fall in $X$ is less than or equal to $\cN$:
	\begin{equation*}
		\sum_{n \in [N]} \mathds{1}_\cS \left( \ds_n, \du_n \right) \leq \cN.
	\end{equation*}
\end{definition}

\begin{remark}
	For simplicity of exposition we restrict attention to counting sets in the form of a product of a set $\cS^{\mathcal Q}$ in state space and a set $\cS^{\mathcal U}$ in mode space. However, our solution method in Section \ref{sec:solving_discrete} can handle more general subsets of $\mathcal Q \times \mathcal U$.
\end{remark}

This notion of constraint extends that in \cite{no_hscc2016} by allowing $\cS$ to be a subset of $\mathcal Q \times \mathcal U$, instead of restricting it to a singleton subset of $\mathcal U$. This generalization permits counting constraints that include the state space, as opposed to only counting the number of systems using each input (i.e., \emph{mode}-counting). 

\begin{example}
  The TCL scheduling constraints \eqref{eq:tcl_1} and \eqref{eq:tcl_2} fit into this class. Let the output space be equal to the state space, i.e. $\mathcal Y = \mathbb{R}$ and $Y(x) = x$. Then the following counting constraints specify that the number of TCLs that are in mode $\on$ at time $t$ should be in the interval $[\munderbar K, \moverbar K]$:
  \begin{equation}
  \label{eq:tcl_mc}
  \begin{aligned}
    & \sum_{n \in [N]} \mathds{1}_{\mathbb{R} \times \{ \on \}} \left( \xc_n(t), \u_n(t) \right) \leq \moverbar K, \\
    & \sum_{n \in [N]} \mathds{1}_{\mathbb{R} \times \{ \off \}} \left( \xc_n(t), \u_n(t) \right) \leq N - \munderbar K.
  \end{aligned}
  \end{equation}
  In addition, the dead band constraints \eqref{eq:tcl_1} can be imposed by the counting constraint
  \begin{equation}
  \label{eq:tcl_sc}
    \sum_{n \in [N]} \mathds{1}_{[\munderbar a, \moverbar a]^C \times \{ \on, \off \}} \left( \xc_n(t), \u_n(t) \right) \leq 0.
  \end{equation}
\end{example}

The goal of this work is to synthesize controllers that satisfy counting constraints for a collection of switched systems with states $\xc_i(t)$, $i \in [N]$, governed by dynamics
\begin{equation}
\label{eq:cont_dynamics}
  \frac{\mathrm{d}}{\mathrm{d}t} \xc_i(t) = f_{\u_i(t)} (\xc_i(t), \dc_i(t)), \;  \u_i : \mathbb{R} \rightarrow \mathcal U,
\end{equation}
for $\xc_i(t) \in \mathbb{R}^{n_x}$ and a disturbance signal $\dc_i(t)$ assumed to take values in a set $\mathcal D$. Although the system description is identical for all $N$ systems in the family, the disturbance can model mild heterogeneity such as modeling inaccuracies and parameter variations (e.g., $\dc_i(t) = \dc_i)$ across subsystems. 

We assume that the vector fields and the disturbance signals satisfy standard assumptions for existence and uniqueness of solutions. In addition, we assume that the vector fields exhibit a certain form of stability \cite{Angeli:2002eb}.
\begin{assumption}
\label{ass:iss}
  For each $\du \in \mathcal U$, the vector field $f_\du(x, d)$ is $C^1$ in $x$ and continuous in $d$. Furthermore, the nominal system $\frac{\mathrm{d}}{\mathrm{d}t} \xc = f_\du(\xc, \mathbf{0})$ is forward-complete and incrementally stable. That is, there exists a $\mathcal {KL}$-function $\beta_\du$ such that
\begin{equation}
\label{eq:kl_func_dyn}
  \| \phi_\du(t,x,\mathbf{0}) - \phi_\du(t,y,\mathbf{0}) \|_\infty \leq \beta_\du \left( \| x - y \|_\infty, t  \right).
\end{equation}
\end{assumption}
Note that the vector field $f_\du(x, d)$ being $C^1$ in $x$ implies that $f_\du(x, d)$ is Lipschitz in $x$ (with some Lipschitz constant $K_\du$) when the states are constrained to a compact set.  

Furthermore, we assume that the disturbance signal is continuous, and that its effect is bounded in an absolute sense.
\begin{assumption}
\label{ass:disturbance}
The disturbance signal $\dc : \mathbb{R} \rightarrow \mathcal D$ is continuous. Furthermore for all $\du \in \mathcal U$, compared to the nominal vector fields without disturbance, the effect of the disturbance is less than $\bar \delta_\du$:
\begin{equation*}
  \| f_\du (x, d) - f_\du(x,0) \|_\infty \leq \bar \delta_\du, \quad \forall d \in \mathcal D.
\end{equation*}
\end{assumption}

To leverage the notion of bisimilarity we consider the time-sampled counterpart of \eqref{eq:cont_dynamics}, and define the problem we want to solve in terms of the resulting transition system rather than the ODE. An ODE solution where state counting constraints are violated in between samplings (but satisfied at sample instants) is still a valid solution to a time-sampled counting problem. If such inter-sample violations are unacceptable, the counting sets can be contracted by some margin determined by the dynamics to ensure satisfaction for all $t \in \mathbb{R}_+$ \cite{liu2016finite}.

We proceed with a problem definition for transition systems. Consider a family of $N$ identical subsystems with dynamics given by a transition system and some initial conditions $\{ \xc_n(0) \}_{n \in [N]}$. The problem we seek to solve is the following:
\begin{problem}
\label{prob:cont_state}
	Consider $N$ subsystems, all governed by the same transition system $\Sigma_{TS} = (\mathcal Q, \mathcal U, \longrightarrow, Y)$, where $\mathcal U$ is finite, and assume that initial states $\{\xc_n(0)\}_{n \in [N]}$ and $L$ counting constraints $\{\cS_l, \cN_l\}_{l \in [L]}$ are given. Synthesize individual switching protocols $\{ \pi_n \}_{n \in [N]}$ such that the generated actions $\u_n(0) \u_n(1) \u_n(2) \ldots$ and trajectories $\xc_n(0) \xc_n(1) \xc_n(2) \ldots$ for $n \in [N]$ satisfy the counting constraints
	\begin{equation}
	\label{eq:prob_counting_constr}
		\sum_{n \in [N]} \mathds{1}_{\cS_l} \left( \xc_n(s), \u_n(s) \right) \leq \cN_l, \quad \forall s \in \mathbb{N}, \; \forall l \in [L].
	\end{equation}
	Let $( N, \Sigma_{TS}, \{ \xc_n(0) \}_{n \in [N]} ,\{\cS_l, \cN_l\}_{l \in [L]})$ denote an instance of this problem.
\end{problem}

Na\"ively, if $\mathcal Q$ is a $n_\ds$-dimensional space, the overall system can be viewed as a $N \times n_\ds$-dimensional hybrid system with $|\mathcal U|^N$ modes and is thus beyond the reach of standard synthesis techniques for large $N$. However, from a specification point of view it does not matter which subsystems that contribute to the summation in \eqref{eq:prob_counting_constr}, so Problem \ref{prob:cont_state} exhibits subsystem permutation symmetry in its specification. In addition, there is also subsystem permutation symmetry in the dynamics since all subsystem states are governed by the same transition system. As an implication there is no need to keep track of identities of individual subsystems---they are all equivalent from both a dynamics- as well as from a specification-point of view. This symmetry (i.e., permutation invariance) is the key to solving large-scale counting problems.


In the following, we construct the time-sampled analogue of \eqref{eq:cont_dynamics} in order to leverage the notion of bisimilarity. We also restrict the domain to a bounded set $\mathcal X \subset \mathbb{R}^{n_x}$ and let the $\tau$-sampled counterpart of \eqref{eq:cont_dynamics} confined to $\mathcal X$ be
\begin{equation}
\label{eq:cont_sampled_dynamics}
	\Sigma_\tau = (\mathcal X, \mathcal U, \underset{\tau}{\longrightarrow}, \text{Id}_{\mathbb{R}^d}),
\end{equation}
where $x \overset{\du}{\underset{\tau}{\longrightarrow}} x'$ if and only if there exists $\dc : [0,\tau] \rightarrow \mathcal D$ such that $x' = \phi_\du(\tau, x, \dc)$. The instance 
\begin{equation}
	\label{eq:continuous_instance}
	( N, \Sigma_\tau, \{ \xc_n(0) \}_{n \in [N]} ,\{\cS_l, \cN_l\}_{l \in [L]})
\end{equation}
of Problem \ref{prob:cont_state} is referred to as the \emph{continuous-state counting problem}. Since $\Sigma_\tau$ is a non-deterministic, infinite transition system, this is a difficult problem to solve. In the following we construct a DFTS $\Sigma_{\tau, \eta}$ that is approximately bisimilar to $\Sigma_\tau$, and relate the corresponding problem instances.

\section{Abstracting the Continuous-State Counting Problem}
\label{sec:cont_to_discrete}

We first recall the abstraction procedure from \cite{Pola:2009bg} to create a finite-state model of \eqref{eq:cont_dynamics}. Under the assumptions outlined above, we establish approximate bisimilarity between the continuous-state system and its finite-state abstraction, which enables us to state results in Section \ref{sub:solvability_of_mode_counting_problem_on_abstraction} relating the solvability of the corresponding counting problems.

\subsection{Abstraction Procedure}
\label{sub:abstraction_procedure}

Consider a system of the form \eqref{eq:cont_dynamics}. For a state discretization parameter $\eta > 0$ we define an abstraction function $\abstr_\eta : \mathcal X \mapsto \mathcal X$ as $\abstr_\eta(x) = \eta \cdot \left\lfloor x / \eta \right \rfloor + (\eta / 2) \mathbf{1}$. This function is constant on hyperboxes of side $\eta$, and its image of the compact set $\mathcal X$ is finite. The abstraction function defines the $(\tau,\eta)$-discretized counterpart of \eqref{eq:cont_dynamics} as the transition system
\begin{equation*}
	\Sigma_{\tau, \eta} =  \left( \abstr_\eta(\mathcal X), \mathcal U, \underset{\tau, \eta}{\longrightarrow},\text{Id}_{\mathbb{R}^d} \right),
\end{equation*}
where $\ds \overset{\du}{\underset{\tau, \eta}{\longrightarrow} }  \ds'$ if and only if $\abstr_\eta (\phi_\du(\tau, \ds, \mathbf{0})) = \ds'$.

The domain $\mathcal X$ is partitioned into uniform boxes of size $\eta$ that represent discrete states, and transitions are established by simulating each mode, without disturbance, during a time $\tau$, starting at the center of the boxes. As opposed to $\Sigma_\tau$, the resulting transition system $\Sigma_{\tau, \eta}$ is finite and deterministic---for each state $\ds$ and action $\du$ there exists (at most) one successor state $\ds'$. This makes the resulting abstract counting problem easier to solve.
Similarly to a result from \cite{Pola:2009bg} we can now show that an abstraction constructed in this way is bisimilar to the time-sampled system \eqref{eq:cont_sampled_dynamics} if a certain inequality holds.
\begin{proposition}
\label{prop:bisimilarity}
	Assume that Assumptions \ref{ass:iss} and \ref{ass:disturbance} hold, i.e., for all $\du \in \mathcal U$ there are Lipschitz constants $K_\du$, $\mathcal{KL}$-functions $\beta_\du$, and disturbance effect bounds $\moverbar \delta_\du$ associated with the modes of \eqref{eq:cont_dynamics}. Then, if for all $\du \in \mathcal U$,
	\begin{equation}
	\label{eq:bisim_ineq}
		\beta_\du(\epsilon, \tau) + \frac{\moverbar \delta_\du}{K_\du} \left( e^{K_\du \tau}-1 \right) + \frac{\eta}{2} \leq \epsilon,
	\end{equation}
	the $(\tau,\eta)$-discretized abstraction $\Sigma_{\tau, \eta}$ and the $\tau$-sampled system $\Sigma_\tau$ are $\epsilon$-approximately bisimilar.
\end{proposition}

The trajectories of $\epsilon$-approximately bisimilar systems remain within distance $\epsilon$ of each other \cite{Girard:2007we}, and the additional term $\frac{\moverbar \delta_\du}{K_\du} \left( e^{K_\du \tau}-1 \right)$ in \eqref{eq:bisim_ineq} introduces a robustness margin that allows the same abstraction to be bisimilar to a family of mildly heterogeneous subsystems. Using these facts we can establish relations between existence of solutions of the counting problem in the continuous-state and discrete-state settings.

\subsection{Relations Between the Continuous-State and Discrete-State Counting Problems}
\label{sub:solvability_of_mode_counting_problem_on_abstraction}

If $\epsilon$-approximate bisimilarity holds, trajectories of $\Sigma_{\tau, \eta}$ and $\Sigma_\tau$ are guaranteed to remain $\epsilon$-close when initial states are $\epsilon$-close and corresponding actions are chosen. Therefore, solvability of the counting problem is equivalent for the two up to an approximation margin $\epsilon$. To make precise statements we introduce functions $\mathcal G^{\pm \epsilon}$ that expand (resp. contract) a counting set $\cS = \cS^{\mathcal X} \times \cS^{\mathcal U}$ in $\mathcal X$-space before quantization:
\begin{equation*}
\begin{aligned}
	\mathcal G^{+\epsilon}(\cS^{\mathcal X} \times \cS^{\mathcal U}) = \abstr_\eta \left( \cS^{\mathcal X} \oplus \mathcal B_\infty(0, \epsilon) \right) \times \cS^{\mathcal U}, \\
	\mathcal G^{-\epsilon}(\cS^{\mathcal X} \times \cS^{\mathcal U}) = \abstr_\eta \left( \cS^{\mathcal X} \ominus \mathcal B_\infty(0, \epsilon) \right) \times \cS^{\mathcal U}.
\end{aligned}
\end{equation*}
Now we can describe how solutions of the counting problem can be mapped between the time-sampled system $\Sigma_\tau$ and the time-sampled and state-quantized system $\Sigma_{\tau, \eta}$.

\begin{theorem}
\label{thm:abs_to_cont}
	Let $\Sigma_{\tau}$ and $\Sigma_{\tau, \eta}$ be the time-sampled and time-sampled and state-quantized systems constructed from a system on the form \eqref{eq:cont_dynamics}, such that $\Sigma_\tau$ and $\Sigma_{\tau, \eta}$ are $\epsilon$-approximately bisimilar. Let $\abstr_\eta$ be the abstraction function for $\Sigma_{\tau, \eta}$.

	If there exists a solution to the instance 
	\begin{equation}
	\label{eq:thm_graph_inst}
	\left(N, \Sigma_{\tau, \eta}, \{ \abstr_\eta(\xc_n(0)) \}_{n \in [N]}, \{  \mathcal G^{+\epsilon}(\cS_l) , \cN_l \}_{l \in [L]} \right)
	\end{equation}
	of Problem \ref{prob:cont_state}, then there exists a solution to the instance 
	\begin{equation}
	\label{eq:thm_cont_inst}
	 (N, \Sigma_\tau, \{ \xc_n(0) \}_{n \in [N]}, \{ \cS_l, \cN_l \}_{l \in [L]} ).
	\end{equation} 
\end{theorem}
An converse result can also be obtained, with the only difference that an additional correction $\eta/2$ is required to account for counting sets that are inflated by quantization.
\begin{theorem}
\label{thm:cont_to_abs}
	Under the same assumptions as in Theorem \ref{thm:abs_to_cont}, if there is no solution to the instance
	\begin{equation}
	\label{eq:thm2_graph_inst}
		\left ( \hspace{-1mm} N, \Sigma_{\tau, \eta}, \{ \abstr_\eta(\xc_n(0)) \}_{n \in [N]}, \hspace{-1mm} \left\{ \mathcal G^{-(\epsilon + \frac{\eta}{2})}(\cS_l), \cN_l\hspace{-0.5mm} \right\}_{l \in [L]} \hspace{-0.5mm}  \right )
	\end{equation}
	of Problem \ref{prob:cont_state}, then there is no solution to the instance 
	\begin{equation}
	\label{eq:thm2_cont_inst}
	 	\left(N, \Sigma_\tau, \{ \xc_n(0) \}_{n \in [N]}, \{ \cS_l, \cN_l \}_{l \in [L]} \right).
	\end{equation}
 \end{theorem}

\begin{example}[continued]
	We proceed with the TCL scheduling problem by adjusting the constraints. The mode-counting constraints \eqref{eq:tcl_mc} lack a state space part, and are therefore not modified. However, by Theorem \ref{thm:abs_to_cont} the local safety constraints \eqref{eq:tcl_sc} must be expanded to
	\begin{equation*}
		\sum_{n \in [N]} \mathds{1}_{[\munderbar a + \epsilon, \moverbar a - \epsilon]^C \times \{ \on, \off \}} \left( \xd_n(s), \u_n(s) \right) \leq 0,
	\end{equation*}
	in order for a solution of the discrete-state instance to be mappable to a valid solution of the continuous-state instance of Problem \ref{prob:cont_state}. Similarly, by Theorem \ref{thm:cont_to_abs}, if it can be shown that the discrete-state instance obtained by contracting the counting set to $[\munderbar a - \epsilon - \eta/2, \moverbar a + \epsilon + \eta/2]^C$ before discretizing is infeasible, then also the original problem is infeasible.
\end{example}

\section{Solving the Discrete-State Counting Problem}
\label{sec:solving_discrete}

Having described the reduction of a continuous-state instance of Problem \ref{prob:cont_state} to a deterministic discrete-state instance, we proceed with a solution procedure for the latter. Consider a DFTS $\Sigma_{DFTS} = (\mathcal Q, \mathcal U, \longrightarrow, Y)$; it may be the result of a continuous-state abstraction or just represent a discrete structure onto which a counting problem is defined. A DFTS can alternatively be viewed as a directed graph $G = (\mathcal Q, \longrightarrow)$, where $\mathcal Q$ is a set of nodes and $\longrightarrow$ is a set of edges, and each edge is labeled with an action from $\mathcal U$. This dual viewpoint allows us to leverage graph-theoretical concepts to investigate the aggregate system; in the following we use the transition system viewpoint and the graph viewpoint interchangeably. We first define aggregate dynamics which take the form of a constrained linear system. The solution method is then presented in the form of a linear feasibility problem over aggregate states.

\subsection{Aggregate Dynamics as a Linear System} 
\label{sub:aggregate_dynamics_as_a_linear_system}

Consider a total of $N$ subsystems whose dynamics are governed by a DFTS $\Sigma_{DFTS} = (\mathcal Q, \mathcal U, \longrightarrow, Y)$.

We introduce $|\mathcal Q|$ aggregate states labeled $w_\ds$ for $\ds \in \mathcal Q$, that describe the number of individual systems that are \emph{at state $\ds$}. By also introducing control actions $r_\ds^\du$ that represent the number of systems \emph{at state $\ds$ using action $\du$}, the aggregate dynamics can be written as
\begin{equation}
	\label{eq:aggregate_dynamics}
	w_\ds(s+1) = \sum_{\du \in \mathcal U} \sum_{\ds' \in \mathcal N_\ds^\du} r_{\ds'}^\du(s), \quad \ds \in \mathcal Q,
\end{equation}
where $\mathcal N_\ds^\du = \left\{ \ds' \in \mathcal Q : \ds' \overset{\du}{\longrightarrow}  \ds \right\}$ is the set of predecessors of $\ds$ under the action $\du$. We constrain the control actions $r_\ds^\du$ such that for all $\du \in \mathcal U$,
\begin{equation}
	\label{eq:cst_pos}
	r_\ds^\du(s) \geq 0, \quad  \sum_{\du \in \mathcal U} r_\ds^\du(s) = w_\ds(s),
\end{equation}
which ensures the continued positivity of the states: $w_\ds (s+1) \geq 0$ for $\ds \in \mathcal Q$. Furthermore, the invariant $\| \mathbf{w}(s) \|_1 = N$ holds over time, where $N$ is the total number of subsystems. In the following, we use the compact notation 
\begin{equation}
\label{eq:agg_dyn_compact}
  \Gamma: \; \mathbf{w}(s+1) = B \mathbf{r}(s),
\end{equation}
to denote this system, where $B$ is composed of the incidence matrices $A^\du, \du \in \mathcal U$ of the system graph. The state space $\mathcal W$ and (state-dependent) admissible control set $\mathcal R$ of this system are then
\begin{equation*}
\begin{aligned}
	\mathcal W & = \left\{ \mathbf{w} \in \mathbb{N}^{|\mathcal Q|} : \| \mathbf{w} \|_1 = N \right\},\\
	\mathcal R(\mathbf{w}) & = \left\{ \mathbf{r} \in \mathbb{N}^{|\mathcal Q||\mathcal U|} : \; \text{\eqref{eq:cst_pos} holds for $(\mathbf{w}, \mathbf{r})$} \right\}.
\end{aligned}	
\end{equation*}

\begin{figure}
	\begin{center}
    \footnotesize
		\input{figures/histogram}
	\end{center}
  \vspace{-4mm}
	\caption{Illustration of how the aggregate dynamics can be interpreted as a time-varying histogram. The continuous state space $\mathcal X$ is partitioned into abstract states $q_0, \ldots, q_7$, and the aggregate states $w_q$ count the number of subsystems (black dots) present in each abstract state. In this example, $w_{q_0} = 1$ and $w_{q_4} = 4$, since there is one subsystem in $q_0$ and four subsystems in $q_4$. As time evolves, the subsystems move around in $\mathcal X$, and the aggregate states follow the dynamics \eqref{eq:aggregate_dynamics}.}
	\label{fig:histogram}
\end{figure}
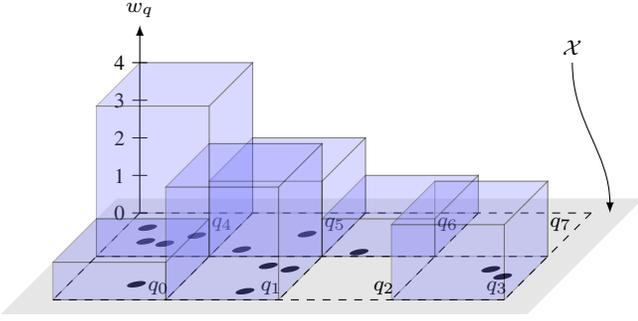

When the discrete-state counting problem is obtained via the abstraction procedure in Section \ref{sec:cont_to_discrete}, the aggregate dynamics \eqref{eq:aggregate_dynamics} can be seen as a time-varying histogram, where each bin in the histogram represents the number of subsystems in a given area of the state space, as illustrated in Fig. \ref{fig:histogram}.

Counting constraints become linear constraints on these aggregate states, so our problem can be rephrased as a linear feasibility problem. However, the feasibility problem is infinite-dimensional since we are searching for solution-trajectories over an infinite horizon. To limit the search to a finite-dimensional space we consider solutions that consist of a finite prefix trajectory and a periodic suffix trajectory. While the prefix trajectory is just a trajectory of the aggregate dynamics, we consider suffix trajectories that are defined in terms of \emph{cycle assignments}.

\subsection{Cycle assignments}
\label{sub:graph_assignments}
 
We consider cycles $C$ of the form 
\begin{equation}
\label{eq:cycle_augmented}
	\cycleC = (\ds_{0}, \du_{0}, \ds_{1}) (\ds_{1}, \du_{1}, \ds_{2}) \ldots (\ds_{{I-1}}, \du_{{I-1}}, \ds_{0}),
\end{equation} 
of length given by $| \cycleC | = I$, where $(\ds_{i}, \du_{i}, \ds_{{i+1}}) \in \longrightarrow$. A \emph{cycle assignment} assigns ``weights'' (or numbers) of subsystems to each node along a cycle. If the graph represents an abstraction, the assignment counts the number of subsystems whose continuous states are in the vicinity of the abstract states on this cycle.
\begin{definition}
	An \textbf{assignment} to a cycle $\cycleC$ is a function $\assgn : [|\cycleC|] \rightarrow \mathbb{R}_+$. If $\assgn(i) \in \mathbb{N}$ for  all $i \in [|C|]$, $\assgn$ is an \textbf{integer assignment}.
\end{definition}
For now we will only be concerned with integer assignments that represent a number of subsystems. However, later in the paper we will study relaxed non-integer solutions consisting of assignments that are not necessarily integral.

An assignment assigns subsystems to a cycle, so that the state of a subsystem circulates along the abstract states on this cycle as time progresses (provided that the appropriate actions are chosen). The movement corresponds to a circular shift of the assignment.
\begin{definition}
	For an assignment $\assgn$ we define its $s$-step \textbf{circulation}, denoted $\assgn^{\circlearrowleft s} : [|\cycleC|] \rightarrow \mathbb{R}_+$, as the shifted function
	\begin{equation*}
		\assgn^{\circlearrowleft s}(i) = \assgn \left( (i - s) \mod |C| \right).
	\end{equation*}
\end{definition}
The periodicity is manifested by the relation $\assgn^{\circlearrowleft (|C| + s)} = \assgn^{\circlearrowleft s}$. To capture how counting quantities vary during the circulation we introduce the following notation for the matching of a cycle and a circulated assignment. 

\begin{definition}
	For a cycle $\cycleC$ and an assignment $\assgn : [|\cycleC|] \rightarrow \mathbb{R}_+$ the \textbf{$\cS$-count at time $s$} of a counting set $\cS$ is defined as
	\begin{equation}
	\label{eq:inner_prod}
		\left\langle \cycleC, \assgn^{\circlearrowleft s} \right\rangle^\cS = \sum_{i \in [|C|]} \begin{cases}
			\alpha^{\circlearrowleft s} (i) \quad & \text{if} \; (q_i, \mu_i) \in X, \\
			0 & \text{otherwise}.
		\end{cases}
	\end{equation}
\end{definition}
If the cycle $\cycleC$ includes elements from a set $\cS$ coming from a counting constraint, \eqref{eq:inner_prod} counts the number of subsystems contributing to the counting constraint at time $s$, assuming that the subsystems make up the assignment $\alpha$ at time zero and circulate in the cycle $C$. The concept is illustrated in Fig. \ref{fig:circulate}.

\begin{figure}
	\begin{center}
		\footnotesize
		\input{figures/assignment_circ0} \hfill \input{figures/assignment_circ1}
	\end{center}
  \vspace{-4mm}
  \caption{Illustration of how an assignment $\assgn = [6,5,4,3,2]$ circulating in a cycle $C = (q_0, \mu_0, q_1) \ldots (q_4, \mu_4, q_0)$ contributes to a counting set $X = \{ (q_1, \mu_1), (q_2, \mu_2), (q_3, \mu_3) \}$. When the assignment is unshifted (left), the $X$-count is $\left\langle C, \assgn \right\rangle^X = 5+4+3 = 12$, and when the assignment has circulated one step, the $X$-count becomes $\left\langle C, \assgn^{\circlearrowleft 1} \right\rangle^X = 6+5+4 = 15$.}
  \label{fig:circulate}
\end{figure}
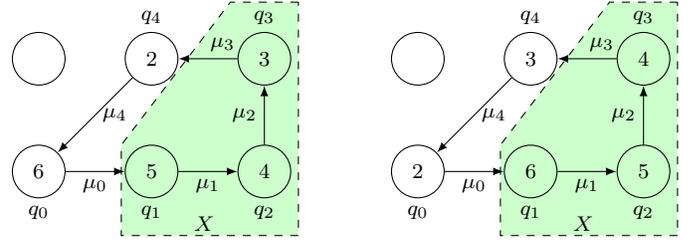

Next we introduce a function that returns the highest count over all circulations for given cycle-assignment pairs.

\begin{definition}
	The \textbf{maximal $\cS$-count} for a cycle $\cycleC$ with assignment $\assgn$, denoted $\maxcnt^\cS(\cycleC, \alpha)$, is the maximal number of subsystems simultaneously in $\cS$ when $\assgn$ is circulated around $\cycleC$:
	\begin{equation*}
		\maxcnt^\cS \left(\cycleC, \assgn \right) 
		= \max_{s \in \mathbb{N}} \left\langle C, \alpha^{\circlearrowleft s} \right\rangle^\cS.
	\end{equation*}
\end{definition}
The maximal $\cS$-count can be computed as the maximum entry in a matrix-vector product $\maxcnt^{\cS}(\cycleC, \alpha) = \left\| B_\cycleC^\cS \alpha \right\|_
	\infty$,
where $B_\cycleC^\cS \in \{ 0,1 \}^{|\cycleC| \times |\cycleC|}$ is a circulant binary matrix. 
\begin{example}
	Going back to cycle and assignment in Fig. \ref{fig:circulate}, it can be seen that the maximal $\cS$-count occurs for $s = 1$, and that $\maxcnt^\cS(C, \alpha) = 15$. The same value is obtained as the infinity norm of the matrix-vector product
	\begin{equation*} 
		\| B_C^\cS \alpha \|_\infty =  \left\|
		\begin{bmatrix}
			0 & 1 & 1 & 1 & 0 \\
			1 & 1 & 1 & 0 & 0 \\
			1 & 1 & 0 & 0 & 1 \\
			1 & 0 & 0 & 1 & 1 \\
			0 & 0 & 1 & 1 & 1
		\end{bmatrix} \begin{bmatrix}
			6 \\ 5 \\ 4 \\ 3 \\ 2
		\end{bmatrix} \right\|_\infty \hspace{-3mm} = \left\| \begin{bmatrix}
			12 \\ 15 \\ 13 \\ 11 \\ 9
		\end{bmatrix} \right\|_\infty \hspace{-3mm} = 15,
	\end{equation*}
  where the first row of the matrix has ones at positions in $C$ corresponding to states in $X$ (i.e. $q_1$, $q_2$, and $q_3$), and the remaining rows are shifted to form a circulant matrix. 
\end{example}
Our solution method to the counting problems requires consideration of multiple assignments that circulate in different cycles. For this reason we introduce the \emph{joint} maximal count of assignments to cycles.
\begin{definition}
\label{def:joint_count}
	The \textbf{maximal joint $\cS$-count} for a set of cycles $\{ \cycleC_j \}_{j \in J}$ and a matching set of assignments $\{ \assgn_j \}_{j \in J}$, denoted $\maxcnt^\cS \left( \{ \cycleC_j \}_{j \in J}, \{ \assgn_j \}_{j \in J} \right)$, is the maximal number of subsystems simultaneously in $\cS$ when the assignments $\{ \assgn_j \}_{j \in J}$ are synchronously circulated around the cycles $\{ \cycleC_j \}_{j \in J}$:
	\begin{equation*}
		\maxcnt^\cS \left( \{ \cycleC_j \}_{j \in J}, \{ \assgn_j \}_{j \in J} \right)
		= \max_{s \in \mathbb{N}} \sum_{j \in J} \left\langle \cycleC_j, \alpha_j^{\circlearrowleft s} \right\rangle^X.
	\end{equation*}
\end{definition}
Also the maximal joint $\cS$-count can be expressed as the maximal element in a sum of matrix-vector products.
\begin{proposition}
	The maximal joint $\cS$-count satisfies
\begin{equation}
\label{eq:joint_count}
\begin{aligned}
	\maxcnt^\cS & \hspace{-1mm} \left( \{ \cycleC_j \}_{j \in J}, \{ \hspace{-0.5mm} \assgn_j \hspace{-0.5mm} \}_{j \in J} \right) \hspace{-0.75mm} = \hspace{-0.75mm} \bigg\| \sum_{j \in J} (\mathbf{1}_{k_j} \hspace{-1.5mm} \otimes\hspace{-1mm}  B_{\cycleC_j}^\cS) \alpha_j \bigg\|_\infty \hspace{-3mm}, 
\end{aligned}
\end{equation}
where $k_j = \lcm \left(\{ |\cycleC_j| \}_{j \in J} \right)/|\cycleC_j|$, $\lcm$ is the least common multiple, $\mathbf{1}_k$ is the length $k$ column vector of all ones, and $\otimes$ is the Kronecker product.
\end{proposition}
The matrices in \eqref{eq:joint_count} are large if the least common multiple is large, which illustrates that relative assignment circulations must be taken into account. We elaborate further on these matters, and how they are connected to graph periodicity, in Section \ref{ssec:size_reduc} below.


\subsection{Solution to the Discrete-State Counting Problem}
\label{sub:solution_to_the_graph_mode_counting_problem}

Consider now a problem instance 
\begin{equation}
\label{eq:graph_instance}
	( N, \Sigma_{DFTS}, \{ \xd_n(0) \}_{n \in [N]} ,\{\cS_l, \cN_l \}_{l \in [L]})
\end{equation}
and let $\mathbf{w}(0)$ be the aggregate initial state
\begin{equation}
\label{eq:aggregate_initial}
  w_\ds(0) = \sum_{n \in [N]} \begin{cases}
    1 & \text{if} \; \xi_n(0) = \ds, \\
    0 & \text{otherwise}.
  \end{cases}
\end{equation}
We restrict attention to solutions of a particular form; it is shown in later sections that this is without loss of generality.
\begin{definition}
	A trajectory for an aggregate initial state $\mathbf{w}(0)$ is of \textbf{prefix-suffix type} if it consists of a finite number of inputs $\mathbf{r}(0), \ldots, \mathbf{r}(T-1)$, and a set of cycles $\{\cycleC_j \}_{j \in J}$ with assignments $\{ \assgn_j \} _{j \in J}$ such that the cycles are populated with their respective cycle assignments at time $T$.
\end{definition}

Given aggregate initial states $\mathbf{w}(0)$, a set $\{ \cycleC_j \}_{j \in [J]}$ of cycles in $G$, and a prefix horizon $T$, prefix-suffix trajectories can be extracted from feasible points of the following linear feasibility problem:
\begin{subequations}
\label{eq:lp}
\begin{align}
	\nonumber \text{find} \; & \assgn_0, \ldots, \assgn_{J-1} \; \text{(cycle assignments)}, & \\
	\nonumber \; & \mathbf{r}(0), \ldots, \mathbf{r}(T-1) \; \text{(aggregate inputs)}, & \\
	\nonumber \; & \mathbf{w}(1), \ldots, \mathbf{w}(T) \; \text{(aggregate states)}, & \\
	\label{eq:lp_c1} \text{s.t.} \; & \sum_{\ds \in \mathcal Q} \sum_{\du \in \mathcal U} \mathds{1}_{\cS_l} \left( \ds, \du \right) r_\ds^\du(s) \leq  \cN_l, & \hspace{-15mm} s \in [T], l \in [L], \\
	\label{eq:lp_c2}  &\maxcnt^{\cS_l}( \left\{ \cycleC_j \right\}_{j \in J}, \left\{ \alpha_j \right\}_{j \in J}) \leq \cN_l, & l \in [L], \\
	\label{eq:lp_c4}  & w_\ds(T) = \sum_{j \in [J]} \left\langle C_j, \alpha_j \right\rangle^{\{ \ds \} \times \mathcal U},  & \ds \in \mathcal Q, \\
	\label{eq:lp_c3}  & \mathbf{w}(s+1) = B \mathbf{r}(s), & s \in [T], \\
	\label{eq:lp_c5} &  \sum_{\du \in \mathcal U} r_\ds^\du(s) = w_\ds(s), & \hspace{-20mm} s \in [T], \ds \in \mathcal Q, \\
	\label{eq:lp_c6} & r_\ds^\du(s) \geq 0, & \hspace{-40mm}s \in [T], \ds \in \mathcal Q, \du \in \mathcal U.
\end{align}
\end{subequations}

The constraints \eqref{eq:lp_c1} and \eqref{eq:lp_c2} enforce counting constraints in the prefix and suffix phases, and \eqref{eq:lp_c4} ensures that the prefix and suffix phases are adequately connected. The remaining constraints \eqref{eq:lp_c3}-\eqref{eq:lp_c6} certify that the prefix trajectory is dynamically feasible with respect to \eqref{eq:aggregate_dynamics}.

The number of variables and (in)equalities in \eqref{eq:lp} are $\mathcal O( T|\mathcal Q||\mathcal U| + \sum_{j \in [J]} |C_j|)$ and $\mathcal O \left( LT + L \lcm \left( \{ |C_j| \}_{j \in J} \right) + T|\mathcal Q| \right)$, respectively (not counting positivity constraints that solvers handle easily). Crucially, these numbers do not depend on the total number of subsystems $N$ which makes the approach suitable for large $N$ and moderate-sized graphs. 

The set $\{ C_j \}_{j \in J}$ of cycles is an input to the optimization problem \eqref{eq:lp}; the solver selects good cycles within this set by setting the assignments of the rest to zero. Section \ref{sec:analysis} presents results regarding when an input cycle set is ``sufficiently rich'' not to compromise the existence of solutions. On the other hand, in practice, we find that randomly sampling a set of cycles is sufficient for feasibility in many problems.

\subsection{Size Reductions of the Linear Program}
\label{ssec:size_reduc}

As we point out next, there are certain ways to further decrease the number of variables and/or inequalities---sometimes without loss of generality. 

First of all, state-mode pairs in $G$ with mandated $0$-count can be pruned from $G$ to decrease the number of aggregate variables. Specifically, we can construct a pruned graph $\tilde G$ in the following way: if $(\ds,\du) \in \cS_l$ and $\cN_l = 0$, then remove the action $\du$ at $\ds$. If this results in nodes in $\tilde G$ with no outgoing edges, prune these nodes together with incoming edges and repeat until all nodes have at least one valid action. This procedure is equivalent to finding the largest controlled-invariant set contained in $\cS_l^C$.

Secondly, the following result drastically reduces the number of suffix counting constraint inequalities in \eqref{eq:lp_c2}.
\begin{proposition}
	The joint $\cS$-count for two cycles $\cycleC_0$ and $\cycleC_1$ with co-prime length, i.e. $\gcd(|\cycleC_0|, |\cycleC_1|) = 1$, can be computed as 
\begin{equation*}
\begin{aligned}
	& \maxcnt^\cS \left( \{ \cycleC_0, \cycleC_1 \}, \{ \assgn_0, \assgn_1 \} \right) \\
	& \; = \maxcnt^\cS \left( \cycleC_0 , \assgn_0 \right) + \maxcnt^\cS \left( \cycleC_1 , \assgn_1 \right).
\end{aligned}
\end{equation*}
\end{proposition}
\begin{proof}
	If $|\cycleC_0|$ and $|\cycleC_1|$ are coprime, by the Chinese Remainder Theorem \cite{rosen2000}, the equations $k_1 = s \mod |C_1|$ and $k_2 = s \mod |C_2|$ have a unique solution $s < |\cycleC_0| |\cycleC_1|$ for every pair $k_1 < |C_1|, k_2 < |C_2|$. It follows that every circulation of $\assgn_0$ in $\cycleC_0$ at some point coincides with every circulation of $\assgn_1$ in $\cycleC_1$; hence the maximal joint $X$-count is equal to the sum of individual maximal $X$-counts for the two cycles.
\end{proof}

The proof illustrates that if there is no periodicity, every relative assignment position will eventually be attained, including the worst-case relative assignment position which is exactly the combined worst-case absolute positions of individual cycles. However, in the presence of periodicity only a subset of all relative positions will be attained and representing the worst-case \emph{within this subset} is not as straight-forward.

More generally, if the cycle set can be partitioned into sets of cycles with mutually co-prime length, the number of inequalities can be reduced. That is, for sets of cycles $\left\{ \cycleC_j \right\}_{j \in J}$ and $\left\{ \cycleC_{j'} \right\}_{j' \in J'}$ with the property that $\text{gcd}(|\cycleC_j|, |\cycleC_{j'}|) = 1$ for all pairs $j \in J$, $j' \in J'$, it holds that 
\begin{equation}
\label{eq:count_sep}
\begin{aligned}
		& \maxcnt^\cS \left( \{ \cycleC_j \}_{j \in J \cup J'}, \{ \assgn_j \}_{j \in J \cup J'} \right) \\
		& \; \; = \maxcnt^\cS  \left( \{ \cycleC_j \}_{j \in J}, \{ \assgn_j \}_{j \in J}  \right) \\
		& \; \; + \maxcnt^\cS  \left(  \{ \cycleC_j \}_{j \in J'}, \{ \assgn_j \}_{j \in J'}  \right).
\end{aligned}
\end{equation}

\begin{example}
	To exemplify this reduction, consider a set of cycles with lengths ranging from 2 to 20. We have $\lcm([21]) = 232792560$ and $\lcm([21]\setminus \{ 11, 13, 17, 19 \}) + 11 + 13 + 17 + 19 = 5100$.
	The number 232792560 is the number of inequalities in the na\"ive approach \eqref{eq:joint_count}, which by \eqref{eq:count_sep} can be drastically reduced to 5100 if the cycles of prime lengths $11,13,17$, and $19$ are considered separately.
\end{example}

If the number of constraints is still prohibitively large, it can be replaced by a conservative constraint as follows.

\begin{remark} 
\label{rem:cyclelength}
	The constraint \eqref{eq:lp_c2} can be substituted by the conservative constraint
	\begin{equation*} 
		\sum_{i \in \mathbb{N}} \maxcnt^{\cS_l}( \{ \cycleC_j \}_{j : |\cycleC_j| = i}, \{ \alpha_j \}_{j : |\cycleC_j| = i} ) \leq \cN_l, \quad \forall l \in [L],
	\end{equation*}
	which groups cycles by cycle length $i$ and disregards effects from periodicity. The number of constraints in this case becomes $L(1 + \| \bigcup_{j \in J} \{ |C_j|\}) \|_1)$ instead of $L \lcm \left( \{ |C_j| \}_{j \in J} \right)$, where $\bigcup_{j \in J} \{ |C_j|\}$ is the set of cycle lengths without repetition.
\end{remark}

While the total number of subsystems $N$ does not impact the number of inequalities or constraints, it might affect the performance of integer linear program solvers since the number of candidate integer points grows with $N$. In addition, the converse results in Section \ref{sec:analysis} depend on $N$---in order to prove infeasiblity of the problem very large horizons $T$ and/or cycle sets may be required. If $N$ is prohibitively large for this purpose it can be scaled down if a certain divisibility condition holds: if there is a common divisor $S$ that divides $w_\ds(0)$ for all $\ds \in \mathcal Q$, and that divides $\cN_l$ for all $l \in [L]$, then there is a 1-1 correspondence between solutions of \eqref{eq:lp} and solutions of its analogue obtained from the substitutions $\mathbf{w}(0) \rightarrow \mathbf{w}(0)/S$ and $\cN_l \rightarrow \cN_l/S$. The correspondence simply consists in scaling $\mathbf{r}$, $\mathbf{w}$, and the assignments $\assgn$ with the same $S$. 

\subsection{Control strategy extraction}

We conclude the section by giving a switching protocol that solves the instance \eqref{eq:graph_instance} from a feasible solution of \eqref{eq:lp}.

\begin{algorithm}
 \label{alg:control_extraction}
 \KwData{Time $s$, current state $\xd_n(s)$ for $n \in [N]$}
 \KwResult{Switching signals $\u_n(s)$ for $n \in [N]$}
 \eIf{$s < T$}{
 	Select $\u_n(s)$ s.t. for all $\ds \in \mathcal Q, \du \in \mathcal U$; $\sum\limits_{n \in [N]} \mathds{1}_{\{ (\ds, \du) \} } \left( \xd_n(s), \u_n(s) \right) = r_\ds^\du(s)$\;
 }{
 	Select $\u_n(s)$ s.t. for all $\ds \in \mathcal Q, \du \in \mathcal U$; $\sum_{n \in [N]} \mathds{1}_{\{ (\ds, \du) \} } \left( \xd_n(s), \u_n(s) \right) = \sum_{j \in [J]} \left\langle C_j, \alpha_j^{\circlearrowleft s-T} \right\rangle^{\{ (\ds, \du) \}}$\;
 }
 \caption{Switching protocol.}
\end{algorithm}

\begin{theorem}
\label{thm:feas_to_sol}
	If $\{ \mathbf{r}(s) \}_{s \in [T]}$, $\{ \mathbf{w}(s) \}_{s \in [T+1]}$, $\{ \alpha_j \}_{j \in [J]}$ is a feasible solution of \eqref{eq:lp}, then input selection according to the switching protocol in Algorithm \ref{alg:control_extraction} is recursively feasible, and solves the instance \eqref{eq:graph_instance}.
\end{theorem}

\begin{remark}
	The switching protocol in Algorithm \ref{alg:control_extraction} depends only on the current states $\{ \xd_n(s) \}_{n \in [N]}$ and on auxiliary information from the solution of \eqref{eq:lp}. However, for implementation central coordination is required at each time step. The coordination requirement can be relaxed by simulating the system up to time $T$ and assigning an individual prefix and suffix to each subsystem. Then decentralized open-loop controllers can be constructed that realize these individual prefix-suffix paths and mimic the performance of the centralized protocol without communication requirements.
\end{remark}

\section{Analysis of the Proposed Linear Program}
\label{sec:analysis}

Theorem \ref{thm:feas_to_sol} establishes that solving \eqref{eq:lp} provides a correct solution to a deterministic discrete instance of Problem \ref{prob:cont_state}. We now discuss completeness of the solution approach and specify the information that can be obtained from (in)feasibility of \eqref{eq:lp} when it is solved as an integer linear program, and when integer constraints are relaxed to obtain more efficiently solvable linear programs. Table \ref{tab:conclusions} summarizes the results of this section.

\begin{table*}
\centering
\caption{Given feasibility or infeasiblity of \eqref{eq:lp} in various configurations, this table lists the inferences that can be made according to results in Section \ref{sec:analysis}.}
\label{tab:conclusions}
\begin{tabular}{llccll}
	\toprule
		\textbf{LP/ILP} & \textbf{Feasible?}  & \textbf{Cycle set} 	&  $T$ 	& \textbf{Result} & \textbf{Why?}   \\
	\midrule 
		ILP 		    & Yes &	 		Any			& 		Any							& Solution  		& Theorem \ref{thm:feas_to_sol}    \\
		LP 				& Yes & 		Any				& 	Any								& Approximate solution & Theorem \ref{thm:rounding_bounds},\ref{thm:nonint_to_int} \\
		ILP 			& No & $\left\{ C : |C| \leq |\mathcal Q| \binom{|\mathcal Q|+N-1}{N}\right\}$ & $\binom{|\mathcal Q|+N-1}{N}$			& No solution exists & Theorem \ref{thm:periodic}  \\
		LP 		 		& No with $\epsilon$-relaxation & $\{ C : C \text{ simple} \}$  	& $\frac{(\diam(G)^2+1)N}{\epsilon}$	& No solution exists & Theorem \ref{thm:sufficient} 	\\
		\bottomrule
\end{tabular}
\end{table*}

\subsection{Graph Properties and Aggregate Controllability}
\label{sub:graph_properties}

We start by connecting properties of the induced graph $G = (\mathcal Q, \longrightarrow)$ to a notion of reachability in the aggregate dynamics \eqref{eq:agg_dyn_compact}. These results are used in the proof of Theorem \ref{thm:sufficient} below, but are also interesting in their own right since they admit a characterization of reachability in the aggregate picture.

We first define a concept of controllability on a subset of nodes $D \subset \mathcal Q$ for the aggregate dynamics $\Gamma$. Similarly as for controllability of linear systems on a subspace, controllability of $\Gamma$ on $D$ means that the system can be steered between any two aggregate states with support on $D$. 

\begin{definition}
  A subset of nodes $D \subset \mathcal Q$ is \textbf{completely controllable} for $\Gamma$ if for any two states $\mathbf{w}$, $\mathbf{w}'$ with support\footnote{A state $\mathbf{w}$ having support on $D$ means that $w_\ds = 0$ for $\ds \not \in D$.} on $D$ such that $\| \mathbf{w} \|_1 = \| \mathbf{w}' \|_1$, there exists a finite horizon $T$, states $\mathbf{w}(s)$, and controls $\mathbf{r} (s)$ satisfying \eqref{eq:cst_pos}, such that $\mathbf{w}(0) = \mathbf{w}$, $\mathbf{w}(T) = \mathbf{w}'$, and $\mathbf{w}(s+1) = B \mathbf{r}(s)$ for $s \in [T]$.
\end{definition}

\begin{theorem}
\label{thm:aperiodic}
  If a strongly connected component $D$ is aperiodic, it is completely controllable for $\Gamma$ in \eqref{eq:agg_dyn_compact}.
\end{theorem}
\begin{proof}
  It is known that the incidence matrix $A_D$ of an aperiodic, strongly connected graph $D$ is primitive \cite{Seneta:2006wd}, i.e., there exists an integer $T$ such that all entries of $A_D^T$ are positive. This means that for each node pair $(\ds_j, \ds_l)$, there exists a path of length $T$ that connects them. Thus, by sending $p_{jl}$ systems along paths $\ds_j \rightarrow \ds_l$ such that $\sum_l p_{jl} = \mathbf{w}_j$ and $\sum_j p_{jl} = \mathbf{w}_l'$, the state at time $T$ is equal to $\mathbf{w}'$. We can define aggregate controls $\mathbf{r}(s)$ that realize these paths by switching the correct number of systems at each node over time.
\end{proof}

In the case of periodicity, it is not possible to reach every state since the parity structure of the initial state is preserved along the trajectories. However, within this restriction, the system is still controllable in a certain sense. If a strongly connected component $D$ has period $P$, its nodes can be labeled with a function $L_P : D \rightarrow [P]$ such that a node $\ds_1$ with $L_P(\ds_1) = p$ only has edges to nodes $\ds_2$ with $L_P(\ds_2) = (p+1) \mod P$. Let $D_0, \ldots, D_{P-1}$ be the subsets of nodes induced by the equivalence relation $\ds_1 \sim \ds_2$ iff $L_P(\ds_1) = L_P(\ds_2)$.

\begin{corollary}
\label{cor:periodicity}
   The subsets of nodes $D_p$ for $p \in [P]$ as constructed above are completely controllable for $\Gamma$.
\end{corollary}
\begin{proof}
  We can connect the nodes in $D_i$ with edges that correspond to paths of length $P$ in $D$. By construction, the resulting graphs are aperiodic, so the previous result applies.
\end{proof}

It is well known that if a discrete-time linear system is completely controllable, its reachable set for a time $s \geq n$, where $n$ is the system dimension, is the entire state-space; otherwise, it is an affine subspace that depends on $s$ and the initial state. The preceding results show a corresponding result for the aggregate dynamics $\Gamma$---a linear system with input constraints evolving on an integer lattice. The controllability in this setting is entirely characterized by the properties of the graph representing the abstraction. For the controllable case (aperiodic graph), the reachable set of $\Gamma$ at a time $s \geq T$, where $T$ is the controllability horizon from the proof of Theorem \ref{thm:aperiodic}, is the set of all positive integer-valued vectors $\mathbf{w}$ that satisfy $\| \mathbf{w} \|_1 = N$. In case of periodicity with a period $P$, it is the intersection of this lattice set with an affine subspace that depends on $(s \mod P)$ and the parity structure of the initial state.

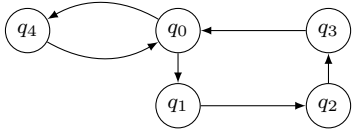
\begin{figure}
  \begin{center}
    \footnotesize
    \begin{tikzpicture}[
      every node/.style = {draw, circle},
      node distance=1.5cm]
      \node (x1) {$\ds_0$};
      \node[below of=x1, node distance=1cm] (x2) {$\ds_1$};
      \node[node distance=2cm, right of=x2] (x3) {$\ds_2$};
      \node[above of=x3, node distance=1cm] (x4) {$\ds_3$};
      \node[node distance=2cm, left of=x1] (x5) {$\ds_4$};
      \draw[-latex] (x1) -- (x2);
      \draw[-latex] (x2) -- (x3);
      \draw[-latex] (x3) -- (x4);
      \draw[-latex] (x4) -- (x1);
      \draw[-latex] (x1) to[bend right] (x5);
      \draw[-latex] (x5) to[bend right] (x1);
    \end{tikzpicture}
  \end{center}
  \caption{A graph with two cycles and period 2. The cycles are $\{ \ds_0, \ds_1, \ds_2, \ds_3 \}$ and $\{ \ds_0, \ds_4 \}$ of length 2 and 4, respectively. Since $\text{gcd}(2,4) = 2$, the period of the graph is 2.}
  \label{fig:per_initialcond}
\end{figure}

\begin{example}
  Consider the graph in Figure \ref{fig:per_initialcond}, it consists of two cycles of length 2 and 4, so the period is 2. The two equivalence classes induced by the periodicity are $\{ \ds_0, \ds_2 \}$ and $\{ \ds_1, \ds_3, \ds_4\}$; all subsystems in the first class move to the second, and vice versa. Consider an initial state such that 
  \begin{equation*}
    \mathbf{w}_{\ds_0}(0) + \mathbf{w}_{\ds_2}(0) = c_0, \quad \mathbf{w}_{\ds_1}(0) + \mathbf{w}_{\ds_3}(0) + \mathbf{w}_{\ds_4}(0) = c_1.
  \end{equation*}
  Due to periodicity, this parity structure is preserved in the sense that
  \begin{equation*}
  \begin{aligned}
      & \mathbf{w}_{\ds_0}(s) + \mathbf{w}_{\ds_2}(s) = \begin{cases}
        c_0, & \text{if $s$ is even}, \\
        c_1, & \text{if $s$ is odd},
      \end{cases} \\
  \end{aligned}
  \end{equation*}
  and conversely for $\mathbf{w}_{\ds_1}(s) + \mathbf{w}_{\ds_3}(s) + \mathbf{w}_{\ds_4}(s)$. However, within this restriction Corollary \ref{cor:periodicity} implies that any assignment is reachable when $s$ is large enough.
\end{example}

\subsection{Converse Results}

The first result states that the restriction to prefix-suffix form is without loss of generality, provided that the prefix horizon is sufficiently large, and that the suffix cycle set is sufficiently rich.

\begin{theorem}
\label{thm:periodic}
	Suppose that there is a solution to the instance \eqref{eq:graph_instance}.
	Then there is a feasible solution to \eqref{eq:lp} with a prefix length $T$ of at most $\binom{|\mathcal Q|+N-1}{N}$ and a suffix consisting of cycles of length at most $|\mathcal Q|\binom{|\mathcal Q|+N-1}{N}$.
\end{theorem}

The upper bounds in Theorem \ref{thm:periodic} yield large feasibility problems; next we present a result that restricts the analysis to much smaller quantities. The key observation is that an assignment can be ``averaged'' over its cycle without violating any counting bounds. The averaging idea is illustrated in Fig. \ref{fig:averaged} and captured in the following definition.

\begin{definition}
For a cycle $\cycleC$, a graph period $P$ dividing $|\cycleC|$, and total weights $N_0, \ldots, N_{P-1}$, the $P$-average assignment $\bar \assgn_{\{ N_p \}_{p \in [P]}}$ is defined as
\begin{equation*}
 	\bar \assgn_{\{ N_p \}_{p \in [P]}}(i) = \frac{N_{(i \mod P)}}{|\cycleC|/P}, \quad \forall i \in [|\cycleC|].
\end{equation*} 
\end{definition}

\begin{figure}
	\begin{center}
		\input{figures/averaging}
	\end{center}
  \vspace{-4mm}
	\caption{Illustration of a non-average assignment $\alpha : [6] \rightarrow \mathbb{N}$, and three averaged assignment with periods 1 (aperiodic), 2, and 3. All assignments have total weight $12$, i.e. $\| \alpha \|_1 = \| \bar \alpha_{12} \|_1= \| \bar \alpha_{\{3,6\}} \|_1= \| \bar \alpha_{\{2,4,6\}} \|_1 = 12$. Note that average assignments are not necessarily integral.}
	\label{fig:averaged}
\end{figure}

In the case $P=1$, this assignment has a constant $\cS$-count for any cycle, more precisely;
\begin{equation}
\label{eq:average_assignment}
\begin{aligned}
	\left\langle C, \bar \alpha_{N_0}^{\circlearrowleft s} \right\rangle^X = 
	& = \frac{N_0}{|\cycleC|} \left\langle C, \mathbf{1} \right\rangle^X
\end{aligned}
\end{equation}
for all $s$, where $\left\langle C, \mathbf{1} \right\rangle^X$ simply counts the number of node-action pairs in $\cycleC$ that are in the counting set ${\cS}$. As a consequence, for any assignment $\assgn$, it holds that
\begin{equation}
	\label{eq:average_vs_other}
	\maxcnt^{\cS} \left(\cycleC, \bar \alpha_{ \| \alpha \|_1 } \right) \leq \maxcnt^{\cS}(\cycleC, \assgn).
\end{equation}
In other words, if the averaged assignment for a given total weight does not satisfy counting bounds, no assignment does. 

A more general result (Lemma \ref{lemma:simple} in the appendix) shows that a cyclic integer suffix can be mapped into a (possibly non-integer) suffix defined on simple cycles via averaging. It turns out that these averaged assignments can be reached by an infinitesimal relaxation of the counting bounds since they preserve the parity structure of the initial condition. The controllability results in Section \ref{sub:graph_properties} do not take counting constraints into account; when such constraints are present they may conflict with controllability. Nevertheless, by introducing an arbitrarily small relaxation of the counting constraints we can ensure that controllability is preserved. The magnitude of the relaxation does however impact the worst-case time required to control the aggregate system to a new state.

\begin{theorem}
\label{thm:sufficient}
	Suppose that there exists an integer solution to the instance $(N, \Sigma_{DFTS}, \{ \xd_n(0) \}_{n \in [N]} ,\{\cS_l, \cN_l \}_{l \in [L]})$. Let $\diam(G)$ be the diameter of the induced graph. Then, if every counting constraint $(\cS_l, \cN_l)$ is relaxed with an absolute factor $\epsilon$ to $(\cS_l, \cN_l + \epsilon)$, the non-integer version of the linear program \eqref{eq:lp} with prefix horizon $\frac{(\diam(G)^2+1)N}{\epsilon}$ and the cycle set consisting of all simple cycles, is feasible.
\end{theorem}

\subsection{Rounding of Non-Integer Solution}
\label{sec:rounding}

If the linear program \eqref{eq:lp} is too large to be solvable as an integer program it may still be possible to solve it using a standard LP solver and round the result to obtain an integer solution. One option is to use a probabilistic discrepancy-minimizing rounding procedure (e.g. \cite{Lovett:2012bt}) which allows specified relationships to be preserved after rounding; thus introducing a counting constraint violation but maintaining the validity of the solution (e.g. dynamics, prefix-suffix connection). Here we instead propose a heuristic to round only the suffix part of the solution and analyze its performance under certain assumptions on cycle structure. For a given integer suffix, the prefix part of \eqref{eq:lp} can be solved to find a matching prefix---a problem that is typically much smaller.

Counting constraints may be violated as a result of the rounding; we give bounds for the magnitude of the worst-case violation. Given an aperiodic graph and (non-integer) cycle-assignment pairs $\{ \cycleC_j, \assgn_j \}_{j \in [J]}$ that satisfy the counting constraints, we propose the following rounding procedure:

\vspace{1mm}

\noindent \textbf{Step 1}: Assign an integer number of subsystems to each cycle that is close to the original weight, i.e., find integers $N_j$ s.t. $\sum_{j \in [J]} N_j = \sum_{j \in [J]} \| \assgn_j \|_1$ and s.t. $N_j$ is close to $\| \alpha_j \|_1$. This can easily be achieved in a way s.t. $\left| N_j - \| \assgn_j \|_1 \right| \leq 1$.

\noindent \textbf{Step 2}: Find individual integer assignments with total weight $N_j$ that are close to the average assignments, i.e., find integer assignments $\tilde \assgn_j$ s.t. $\| \tilde \assgn_j \|_1 = N_j$ and s.t. $\assgn_j$ is close to $\bar \assgn_{N_j}$.

  To this end, we let $\kappa_1$ and $\kappa_2$ be the quotient and remainder when dividing $N_j$ by $|\cycleC_j|$, i.e. $\kappa_1 = \lfloor N_j/|C_j| \rfloor$ and $\kappa_2 = N_j \mod |C_j|$. Then let $d = |\cycleC_j|/\kappa_2$. We consider the pseudo-periodic assignment $\tilde \assgn_j$ defined as follows:
  \begin{equation}
  \label{eq:periodic_assignment}
  \begin{aligned}
      \tilde \assgn_j(i) & =  \kappa_1 + 1, \text{ for } i \in  \left\{ \left \lfloor d k \right \rfloor \right\}_{k \in [\kappa_2]},\\
      \tilde \assgn_j(i) & =  \kappa_1, \text{ otherwise}.
  \end{aligned}
  \end{equation}
  This assignment is pseudo-periodic in the sense that the 1's are evenly distributed with distance $d$ before they are rounded to integer points, as illustrated in Fig. \ref{fig:pseudo-periodic}. It is easy to see that $\| \tilde \assgn_j \|_1 = N_j$.

\noindent \textbf{Step 3}: Find a prefix using \eqref{eq:lp} that steers to $\{ \cycleC_j, \tilde \assgn_j \}_{j \in [J]}$.

\vspace{1mm}




If the graph is periodic the rounding can be done so as to preserve the parity structure of the original solution and guarantee reachability, but the details are omitted here. Assuming aperiodicity, we want to compare the counting bounds for the rounded solution with the counting bounds for the original solution. First we give a result that assumes a certain structure of a cycle, namely that all nodes that contribute to the counting set $X$ are placed in sequence in the cycle. Such structure is often present in practical examples, as connected counting regions tend to lead to consecutive parts in cycles.

\begin{proposition}
\label{prop:pseudo_periodic}
	Let $R_\cycleC^\cS = \left\langle C, \mathbf{1} \right\rangle^X$ be the number of nodes in a cycle $\cycleC$ that contribute to $X$-counting. If all such nodes are consecutive, i.e., $(\ds_{i}, \du_{i}) \in {\cS}$ for $i \in [R_\cycleC^{\cS}]$ and $(\ds_{j}, \du_{j}) \not \in \cS$ for $j \in [|\cycleC|]\setminus[R_\cycleC^{\cS}]$, then the rounding procedure \eqref{eq:periodic_assignment} satisfies
	\begin{equation*}
		\maxcnt^X(\cycleC_j,\tilde \assgn_j) \leq \maxcnt^X(\cycleC_j, \bar \assgn_{N_j}) + 1.
	\end{equation*}
\end{proposition}

\tikzsetnextfilename{pseudo-periodic}
\begin{figure}
	\begin{center}
    \footnotesize
		\input{figures/rounding_left}
	\end{center}
  \vspace{-4mm}
	\caption{Illustration of the pseudo-periodic assignment for $|\cycleC_j|=7$, $\| \assgn_j \|_1 = 3$. The subsystems are first placed in non-integer positions separated by a distance $d$ (open circles), and then rounded into integer positions (filled circles). The final pseudo-periodic assignment is $[1, 0,1, 0, 1, 0, 0]$.}
	\label{fig:pseudo-periodic}
\end{figure}
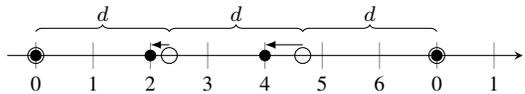

\begin{proof}	
	When the assignment $\tilde \assgn_j$ circulates in $\cycleC_j$, exactly $R_\cycleC^\cS$ contiguous indices of $\tilde \assgn_j$ contribute to the $\cS$-count. We bound the number of contributing indices with value $\kappa_1 + 1$. Let $[i_0, i_0 + R_\cycleC^\cS-1]$ be a (circular) sequence representing $R_\cycleC^\cS$ contributing indices. Consider Fig. \ref{fig:pseudo-periodic}; any point that ends up in the sequence after left-rounding must satisfy $d k \in [i_0, i_0 + R_\cycleC^\cS)$, where $k \in [\kappa_2]$. Since each left-closed, right-open interval of length $d$ captures exactly one point of the form $d k$, there are at most $\left \lceil R_\cycleC^\cS/d \right \rceil$ such points. Therefore by \eqref{eq:average_assignment},
	\begin{equation*}
	\begin{aligned}
	 	 	& \maxcnt^{\cS} (\cycleC_j, \tilde \assgn_j) \leq \kappa_1 R_\cycleC^\cS + \left \lceil \frac{R_\cycleC^\cS}{d} \right \rceil 
	 	 	\leq \kappa_1 R_\cycleC^\cS + \frac{R_\cycleC^\cS}{d} + 1  \\
	 	 	& \quad \leq \frac{R_\cycleC^\cS}{|C_j|} \underbrace{\left( \kappa_1 |C_j| + \kappa_2 \right)}_{N_j} + 1 = \maxcnt^X(\cycleC_j,\bar \assgn_{N_j}) + 1.
	 \end{aligned}
	 \vspace{-11mm}
	 \end{equation*}
\end{proof}

\begin{corollary}
\label{cor:cycle_round}
	If $C_j$ has at most $p$ ${\cS}$-segments (consecutive nodes taking values in $\cS$), the rounding \eqref{eq:periodic_assignment} satisfies $\maxcnt^X(\cycleC_j, \tilde \assgn_j) \leq \maxcnt^X(\cycleC_j,\bar \assgn_{N_j}) + p$.
\end{corollary}
\begin{proof}
	Follows from applying Proposition \ref{prop:pseudo_periodic} to each of the $p$ segments.
\end{proof}

We now incorporate these results into a bound on the counting constraint violation in the overall rounding procedure. Again, this bound does not depend on the total number of subsystems $N$. The difference between the original counting bounds and their relaxed counterparts therefore becomes insignificant as $N$ grows.

\begin{theorem}
\label{thm:rounding_bounds}
For a given counting constraint $({\cS},\cN)$ satisfied by $\{ \alpha_j \}_{j \in [J]}$, the following bound holds for the overall suffix rounding procedure
\begin{equation*}
	\sum_{j \in [J]} \maxcnt^{{\cS}}(\cycleC_j, \tilde \assgn_j) \leq \cN + J +\sum_{j \in [J]} p_j^{{\cS}},
\end{equation*}
where $p_j^{\cS}$ is the number of ${\cS}$-segments in the $j$'th cycle. Thus the relaxed counting constraint $({\cS}, \cN+J+\sum_{j \in [J]} p_j^{{\cS}})$ is guaranteed to be satisfied by the rounded suffix $\{ \tilde \alpha_j \}_{j \in [J]}$.
\end{theorem}
\begin{proof}
	By the rounding procedure, Corollary \ref{cor:cycle_round}, and \eqref{eq:average_vs_other},
	\begin{equation*}
	\begin{aligned}
			& \maxcnt^{\cS}(\cycleC_j, \tilde \assgn_j) \leq \maxcnt^{\cS}(\cycleC_j, \bar \assgn_{N_j}) + p_j^{\cS} \\
			& \leq \maxcnt^{\cS}(\cycleC_j, \bar \assgn_{\| \alpha_j \|_1}) + 1 + p_j^{\cS}. \\
	\end{aligned}
	\end{equation*}
	Noting that $\maxcnt^{\cS}(\cycleC_j, \bar \assgn_{\| \alpha_j \|_1}) \leq \maxcnt^{\cS}(\cycleC_j, \assgn_j)$ and summing over $j \in [J]$ gives the result.
\end{proof}

If the structure required in Proposition \ref{prop:pseudo_periodic} is not present, the following is a worst-case bound on the counting constraint violation due to rounding. 

\begin{theorem}
\label{thm:nonint_to_int}
	The rounding \eqref{eq:periodic_assignment} satisfies
	\begin{equation*}
	\begin{aligned}
			\maxcnt^\cS(\cycleC_j, \tilde \assgn_j) \leq \maxcnt^\cS \left( \cycleC_j, \bar \assgn_{N_j} \right) + \frac{|\cycleC_j|}{4}.
	\end{aligned}
	\end{equation*}
\end{theorem}
\begin{proof}
	We assume that $N_j < |\cycleC_j|$, since any multiple of $|\cycleC_j|$ can be assigned as the average assignment. The number of nodes contributing to the $\cS$-count is upper bounded by $\min(R^X_\cycleC, N_j)$, hence,
	\begin{equation*}
	\begin{aligned}
			& \maxcnt^\cS (\cycleC_j, \tilde \assgn_j) - \maxcnt^\cS \left(\cycleC_j, \bar \assgn_{N_j} \right) \\
			& \; \leq \min(R^X_\cycleC, N_j) - \frac{N_j}{ |\cycleC_j| } R^X_\cycleC \\
			&\;  = |\cycleC_j| \min \left( \frac{R^X_\cycleC}{|\cycleC_j|} \left( 1 \hspace{-1mm} - \hspace{-0.5mm} \frac{N_j}{|\cycleC_j|} \right), \frac{N_j}{|\cycleC_j|} \left(1 \hspace{-1mm} - \hspace{-0.5mm} \frac{R^X_\cycleC}{|\cycleC_j|}  \right) \right) \leq \frac{|\cycleC_j|}{4},
	\end{aligned}
	\end{equation*}
	where the last step follows from $\max_{a,b \in [0,1]} \min (ab, (1-a)(1-b)) = 1/4$.
\end{proof}

These rounding bounds can be precomputed (using all cycles used in the LP instead of all cycles with non-zero assignments) and the counting constraints can be strengthened accordingly to ensure that the original constraints are satisfied after rounding.

\section{Extension to Strong Heterogeneity}
\label{sec:multiclass}

Up until now we considered mild heterogeneity in the continuous dynamics, which allowed us to construct a single abstraction that captures all the possible behaviors. If there is significant heterogeneity in the collection of subsystems, this may no longer be possible while maintaining a good level of approximation.

To alleviate this shortcoming, we can extend our synthesis method to a \emph{multi-class} setting where each subsystem belongs to a particular class, and each class has only mild heterogeneity among its members. One abstraction per class can then be constructed, and the counting problem can be solved jointly for the different classes. In addition, counting constraints can be extended to capture class identity. For instance, we can posit that at least $\cN_1$ subsystems of class 1 be present in a given area, or guarantee that no more than $\cN_2$ subsystems of class 2 are in a particular dynamic mode.

To formalize these ideas, consider $H$ transition systems $\Sigma^h = (Q^h, \mathcal U^h, \longrightarrow ^h, Y^h)$ for $h \in [H]$. Then the multi-class counting problem is as follows:
\begin{problem}
\label{prob:multiclass}
	Consider $N$ subsystems divided in $H$ classes such that class $h$ has $N_h$ members, and $\sum_{h \in [H]} N_h = N$. Subsystems in class $h$ are governed by the transition system $\Sigma^h$. Assume that for all $h \in [H]$, initial states $\{\xc_n^h(0)\}_{n \in [N_h]}$ are given.

	Given $L$ counting constraints $\{\prod_{h \in [H]} \cS_l^h, \cN_l\}_{l \in [L]}$ with counting sets $\cS_l^h \subset \mathcal Q^h \times \mathcal U^h$, synthesize individual switching protocols $\{ \pi_n^h \}_{n \in [N_h]}$ such that the generated actions $\u_n^h(0) \u_n^h(1) \u_n^h(2) \ldots$ and trajectories $\xc_n^h(0) \xc_n^h(1) \xc_n^h(2) \ldots$ satisfy the counting constraints
	\begin{equation*}
		\sum_{h \in [H]} \sum_{n \in [N_h]} \mathds{1}_{\cS_l^h} \left( \xc_n^h(s), \u_n^h(s) \right) \leq \cN_l, \quad \forall s \in \mathbb{N}, \; \forall l \in [L].
	\end{equation*}
\end{problem}

The linear program \eqref{eq:lp} can easily be extended to the multi-class setting, at the cost of additional variables and constraints. Assuming similar abstraction parameters and cycle selections, a problem with two classes has roughly twice as many variables as a problem with a single class. The next section includes an example showcasing how a multi-class counting problem can account for parameter heterogeneity in a family of continuous-time systems.

\section{Examples}
\label{sec:examples}

We showcase the method on two examples; one numerical example and the TCL scheduling problem. The examples are computed with our prototype implementation available at \url{https://github.com/pettni/mode-count}, which uses Gurobi \cite{gurobi} as the underlying ILP solver.

\subsection{Numerical Example} 
\label{sub:numerical_example}

Our first example is the following non-linear system:
\begin{equation}
\label{eq:sat_dynamics}
\begin{aligned}
	\dot \xc_1 & = -2(\xc_1- \mathbf{u}) + \xc_2, \\
	\dot \xc_2 & = -(\xc_1- \mathbf{u})-2\xc_2-\xc_2^3. \\
\end{aligned}
\end{equation}
It can be shown that for a constant $\mathbf{u}$ the system is incrementally stable and that the $\mathcal {KL}$-function
\begin{equation*}
	\beta(r,t) = \sqrt{2} r \left\| \exp \left( \begin{bmatrix}
		-2 & 1 \\ -1 & -2
	\end{bmatrix} t \right) \right\|_2
\end{equation*}
satisfies \eqref{eq:kl_func_dyn}. We consider two modes $\du_1$ and $\du_2$ corresponding to the constant inputs $\mathbf{u} = -1$ and $\mathbf{u} = 1$.

We consider the domain $\mathcal X = \{ (x_1, x_2) : x_1 \in [-2,2], x_2 \in [-1.5, 1.5] \}$. For a large $N$, we introduce mode-counting constraints $(\cS_{\du_{1}}, 0.55 N)$ and $(\cS_{\du_{2}}, 0.55 N)$ for $\cS_{m_s} = \mathcal X \times \{ s \}$ stating that at most 55\% of the subsystems can use the same dynamical mode at any given time. In addition we consider a balancing objective; namely that no more than $55 \%$ of the subsystems should be in one of the sets $\cS_1 = \{ (x_1, x_2) :  x_1 \geq 0 \}$ or $\cS_2 = \{ (x_1, x_2) : x_1 \leq 0 \}$; expressed by the counting constraints $(\cS_1, 0.55 N)$ and $(\cS_2, 0.55 N)$. Furthermore, we want all subsystems to repeatedly visit these two sets.

Using abstraction parameters $\eta = 0.05$, $\tau = 0.32$ we obtain an abstraction with 4,941 states that is $0.1$-approximately bisimilar to the time-discretization of \eqref{eq:sat_dynamics}. In order to guarantee that the counting constraints are satisfied, we therefore need to expand the counting sets as $\tilde \cS_1 =  \{ (x_1, x_2) :  x_1 \geq -0.1 \}$ and $\tilde \cS_2 =  \{ (x_1, x_2) :  x_1 \leq 0.1 \}$. We proceed by solving the discrete counting problem with randomized initial conditions and a horizon $T = 10$. We sampled 200 randomized cycles that visit both $\tilde \cS_1^C$ and $\tilde \cS_2^C$, in order to achieve the second objective. We solved the problem for $N = 10^k$ for $k = 2, \ldots, 9$; the solving times are shown in Table \ref{tab:scaling} and illustrate that the difficulty of this problem is largely independent of $N$.

\begin{table}[h]
	\caption{Average solution times over 10 randomized trials}
	\label{tab:scaling}
	\centering
	\begin{tabular}{l c c c c c c c c c}
		\toprule
		N & $10^2$ & $10^3$ & $10^4$ & $10^5$ & $10^6$ & $10^7$ & $10^8$ & $10^9$ \\
		\midrule
		Time (s) & 7.8 & 11.1 & 12.0 & 12.8 & 11.5 & 13.2 & 13.0 & 11.0  \\
		\bottomrule
	\end{tabular}
\end{table}	

Fig. \ref{fig:sat_counting} illustrates the number of systems that are in the counting sets over time in a trajectory, and Fig. \ref{fig:sat_cycle} demonstrates some of the cycles that make up the suffix part of the solution.

\tikzsetnextfilename{nl_counting}
\begin{figure}
	\begin{center}
		\footnotesize
		\input{figures/nl_counting}
	\end{center}
  \vspace{-5mm}
	\caption{Fraction of total number of systems present in $\tilde \cS_1$, $\tilde \cS_2$ and $\cS_{{\du_1}}$ over time. Whereas the sets $\cS_{\du_1}$ and $\cS_{\du_2}$ are mutually exclusive (the fractions sum to 1), the sets $\tilde \cS_1$ and $\tilde \cS_2$ are not due to the expansion to account for approximate bisimilarity.}
	\label{fig:sat_counting}
\end{figure}
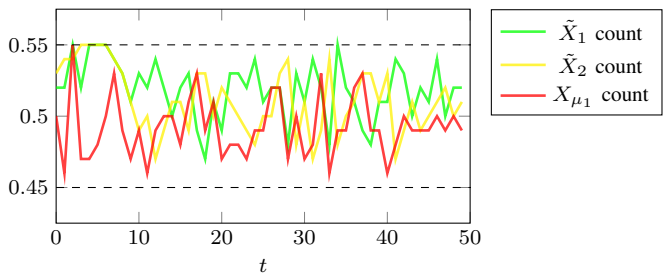

\tikzsetnextfilename{nl_cycles}
\begin{figure}
	\begin{center}
		\footnotesize
		\input{figures/nl_cycles}
	\end{center}
  \vspace{-5mm}
	\caption{Illustration of counting sets $\cS_1$ and $\cS_2$ together with selected cycles making up the suffix solution.}
	\label{fig:sat_cycle}
\end{figure}
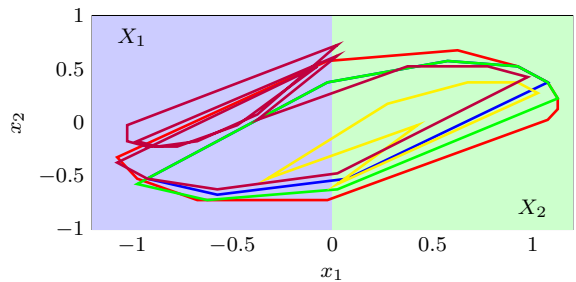

\subsection{Application Example: TCL Scheduling}
\label{sub:application_example}
{}
We use the following model for the dynamics of the temperature $\xc_n$ of an individual TCL \cite{Hao:2015fp}:
\begin{equation}
\label{eq:tcl_dyn}
	\dot \xc_n(t) = -a_n \left( \xc_n(t) - \theta^a_n \right) - b_n P_n^m \mathds{1}_{\{\on\}}  \left( \u_n(t) \right).
\end{equation}
We assume that there are two distinct populations of TCLs, i.e., that all sets of parameters $(a_n, b_n, \theta_n^a, P_n^m)$ are $\moverbar \delta$-close to one of two nominal parameter configurations. Each nominal parameter configuration represents a mildly heterogeneous class (c.f. Section \ref{sec:multiclass}). The parameter values for the nominal configurations are listed in Table \ref{tab:tcl_param} along with the abstraction parameters $\eta$ and $\tau$ used for each class and the allowed deviation $\moverbar \delta$ from these nominal values. These parameters result in finite abstraction with 1,600 and 1,200 states for the two classes, respectively.

\begin{table}[h]
	\center
	\caption{Parameter values for the two classes of subsystems. Note that the time discretizations need to be identical for concurrent execution.}
	\label{tab:tcl_param}
	\begin{tabular}{l c c}
	\toprule
	Parameter 			& Class 1     & Class 2 \\
	\midrule
	Nominal $[a,b,\theta_a, P_m]$& $[2, 2, 32, 5.6]$ & $[2.2, 2.2, 32, 5.9$]\\
	Space discretization $\eta$  & 0.002  & 0.0015 \\
	Time discretization  $\tau$  & 0.05   & 0.05   \\
	Error bound $\moverbar \delta$ & 0.025 & 0.025  \\
	\bottomrule
	\end{tabular}
\end{table}

It can easily be shown that the $\mathcal {KL}$-function $\beta(r,s) = r e^{-s a_n}$ satisfies \eqref{eq:kl_func_dyn} with respect to \eqref{eq:tcl_dyn}. In addition, it satisfies the approximate bisimulation inequality \eqref{eq:bisim_ineq} for an approximation level $\epsilon = 0.2$ and using the Lipschitz constant $a_n$ for \eqref{eq:tcl_dyn}; thus the results from Section \ref{sub:solvability_of_mode_counting_problem_on_abstraction} apply. The constraints for the TCL problem have been introduced earlier in \eqref{eq:tcl_mc}-\eqref{eq:tcl_sc}. We posit that all subsystems must remain in the temperature interval $[21.3, 23.7]$; taking the approximation into account this implies that the constraint for the discrete problem becomes $\xc_n(t) \in [21.5, 23.5]$.

We randomly selected initial conditions for 10,000 systems of each class, sampled 50 random cycles for each class\footnote{To promote diversity in the cycle sets, the cycles were selected in order to have different fractions of time in mode $\on$.}, and solved the ILP \eqref{eq:lp} for a prefix length of 20 steps (corresponding to one hour). We also introduced randomized additive model errors $\delta_n$ such that $|\delta_n| \leq \moverbar \delta$ to represent mild in-class heterogeneity. Figure \ref{fig:tcl_density} shows simulated trajectory densities for two different mode-$\on$-counts, one \emph{maximal} count of 6,000, i.e. $\sum_{n \in [N]} \mathds{1}_{\{ \on \}} \left( \sigma_n(t) \right) \in [0, 6000]$; and one \emph{minimal} count of 6,700, i.e. $\sum_{n \in [N]} \mathds{1}_{\{ \on \}} \left( \sigma_n(t) \right) \in [6700, N]$. For comparison, the fundamental minimal upper bound is 5,595 and the fundamental maximal lower bound is 7,045 as computed from formulas in \cite{no_hscc17} that apply to a centralized full-state feedback coordinator with arbitrarily fast switching. While the objective here was not to find the maximal ranges (which could be done by adding an objective function to \eqref{eq:lp}), the fundamental limits can not be attained due to approximation errors stemming from the approximate bisimulation, incomplete cycle selection, a minimal dwell time imposed by the time discretization, etc. Figure \ref{fig:tcl_count} shows mode-$\on$-counts during the same simulation for the two experiments. As can be seen, the imposed counting bounds are satisfied.


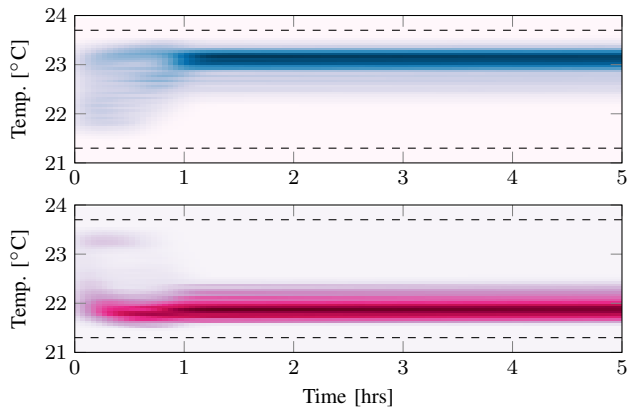
\begin{figure}
	\begin{center}
    \footnotesize
		\tikzsetnextfilename{density1} \input{figures/tcl_density1} ~
		\tikzsetnextfilename{density2} \input{figures/tcl_density2}
	\end{center}
  \vspace{-5mm}
	\caption{Density of TCLs in different parts of the temperature spectrum over time---blue parts represent regions with a larger fraction of the 20,000 subsystems. The state counting constraint \eqref{eq:tcl_sc} guarantees that no subsystem exits the interval [21.3, 23.7] which is marked with dashed lines. The first hour represents the prefix part of the solution which steers the initial state to the periodic suffix.}
	\label{fig:tcl_density}
\end{figure}

\begin{figure}
	\begin{center}
    \footnotesize
		\tikzsetnextfilename{tlc_count} \input{figures/tcl_count}
	\end{center}
  \vspace{-5mm}
	\caption{Number of TCLs in mode $\on$ during the two simulations. As can be seen, the lower bound of 6,700 is enforced for the upper (red) trajectory, while the upper bound 6,000 is enforced for the lower (blue) trajectory.}
	\label{fig:tcl_count}
\end{figure}
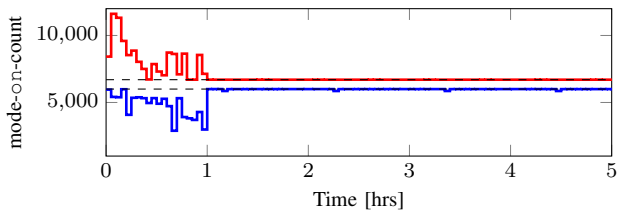

\section{Conclusion}
\label{sec:conclusion}

This paper was concerned with control synthesis for very high-dimensional but permutation-symmetric systems subject to likewise symmetric counting constraints, and presented a scalable sound and (almost) complete solution to this problem. The main insight is to aggregate the individual subsystem dynamics as an integer linear system induced from an abstraction constructed for a single subsystem, thus avoiding abstraction of the entire state space. As we used an approximately bisimilar system as an abstraction, the same abstraction can represent not only identical subsystems (homogeneity) but also subsystems with slightly different parameters (i.e., almost symmetric, or mildly heterogeneous) by a slight change in the approximation factor. The control synthesis problem was then reduced to one of coordinating the number of subsystems that are in different parts of the discrete state-space of this abstraction. We characterized prefix-suffix solutions as the feasible set of an integer linear program, and showed how to interpret (in)feasibility of the program both in the integer and non-integer case.

The results were demonstrated on a TCL scheduling problem including tens of thousands of subsystems. With the proposed approach, it is possible to impose hard constraints on the overall power consumption of TCLs over an infinite time horizon, to the best of our knowledge this is a first in this domain. Counting constraints are also relevant in other application domains, including multi-agent planning and coordination as shown in \cite{sno_2016}. 

Exploiting symmetries to achieve scalability in correct-by-construction methods is a promising direction and we will explore other types of symmetries in our future work. Another interesting direction is to consider other types of abstractions, including non-deterministic ones, since not all systems admit finite bisimulations. Although the idea of an aggregate system can still be used in this case, one should either solve a robust uncertain ILP, which could lead to conservative results, or consider reactive feedback solutions, for which different synthesis techniques should be developed.

\section*{Acknowledgment}

The authors would like to thank Johanna Mathieu for insightful discussions regarding TCL coordination problem. Petter Nilsson is supported by NSF grant CNS-1239037. Necmiye Ozay is supported in part by NSF grants CNS-1446298 and ECCS-1553873, DARPA grant N66001-14-1-4045 and an Early Career Faculty grant from NASA's Space Technology Research Grants Program.

\bibliographystyle{IEEEtran}
\bibliography{IEEEabrv,ref}

\appendix

\section{Proofs}
\label{sec:app_proofs}

\begin{proof}[Proof of Proposition \ref{prop:bisimilarity}]

For a grid point $\ds \in \abstr_\eta(\mathcal X)$ and $x \in \mathcal X$, consider the relation $\ds \sim x$ iff $\| \ds - x \|_\infty \leq \epsilon$. Since $\epsilon > \eta/2$ this relation has the property that every $x \in \mathcal X$ is related to some $q \in \abstr_\eta(\mathcal X)$, and conversely every $q \in \abstr_\eta(\mathcal X)$ is related to at least one $x \in \mathcal X$. The relation satisfies 1) of Definition \ref{def:bisimilar}. To show that also 2) and 3) hold, we show the stronger property that if $\ds \sim x$, then for all $x' = \phi_\du(\tau,x,\dc)$ pertaining to an admissible $\dc : [0, \tau] \rightarrow \mathcal D$ it holds that $\ds' \sim x'$ for (the unique) $\ds'$ such that $\ds \underset{\tau, \eta}{\overset{\mu}{\longrightarrow}} \ds'$.

From the construction of $\Sigma_{\tau, \eta}$, $\| \ds' - \phi_\du(\tau, \ds, \mathbf{0}) \|_\infty \leq \eta/2$. Furthermore, under the continuity assumptions, it is known that the flow $\phi_\du(t, x, \mathbf{0})$ of $f_\du(x,0)$ and the flow $\phi_\du(t, x, \dc)$ for any $\dc(t)$ such that $\max_{t \in [0, \tau]} \|f_\du(x, 0) - f_\du(x, \dc(t)) \|_\infty \leq \moverbar \delta_\du$ satisfy $\| \phi_\du(t, x, \mathbf{0}) - \phi_\du(t, x, \dc) \|_\infty \leq (\moverbar \delta_\du / K_\du) (e^{K_\du t}-1)$ \cite{Hirsch:2012tx}. Thus, it follows that for any such $\dc : [0, \tau] \rightarrow \mathcal D$,
\begin{equation*}
\begin{aligned}
    \| & \ds' - \phi_\du( \tau,  x, \dc) \|_\infty \leq \| \phi_\du(\tau, x, \dc) - \phi_\du(\tau, x, \mathbf{0}) \|_\infty \\
    & + \| \phi_\du(\tau, x, \mathbf{0}) - \phi_\du(\tau, \ds, \mathbf{0}) \|_\infty 
    + \| \phi_\du(\tau, \ds, \mathbf{0}) - \ds' \|_\infty \\
    & \leq \frac{\moverbar \delta_\du}{K_\du} \left( e^{K_\du \tau} - 1 \right) + \beta_\du \left( \| \ds-x \|_\infty, \tau \right) + \frac{\eta}{2}  \\
    & \leq  \frac{\bar \delta_\du}{K_\du} \left( e^{K_\du \tau} - 1 \right) + \beta_\du \left( \epsilon, \tau \right) + \frac{\eta}{2} \leq \epsilon.
\end{aligned}
\end{equation*}
Hence $\ds' \sim \phi_\du( \tau,  x, \dc)$, which completes the proof.
\end{proof}

\begin{proof}[Proof of Theorem \ref{thm:abs_to_cont}]

  Let $\{ \pi_n \}_{n \in [N]}$ be individual switching protocols that solve \eqref{eq:thm_graph_inst} by generating trajectories $\xd_n(0)\xd_n(1)\ldots$ and actions $\u_n(0)\u_n(1)\ldots$ for $\Sigma_{\tau, \eta}$. Due to bisimilarity and the $\eta/2$-proximity of initial conditions, the individual trajectories $\xi_n(0)\xi_n(1)\ldots$ of $\Sigma_{\tau, \eta}$ and the individual trajectories $\xc_n(0) \xc_n(1) \ldots$ of $\Sigma_\tau$ satisfy $\| \xd_n(s) - \xc_n(s) \|_\infty \leq \epsilon$ for all $s \in \mathbb{N}$ when the action sequence $\u_n(0) \u_n(1) \ldots$ is implemented for both systems. By assumption,
  \begin{equation*}
    \sum_{n \in [N]} \mathds{1}_{\mathcal G^{+\epsilon} \left( \cS_l \right)} \left( \xd_n(s), \u_n(s) \right) \leq \cN_l.
  \end{equation*}
  Thus for a counting set $\cS_l = \cS_l^{\mathcal X} \times \cS_l^{\mathcal U}$,
  \begin{equation*}
  \begin{aligned}
      \xc_n(s) \in \cS_l^{\mathcal X} & \implies \xd_n(s) \in \cS_l^{\mathcal X} \oplus \mathcal B_\infty(0, \epsilon) \\
      & \implies \xd_n(s) \in \abstr_\eta \left( \cS_l^{\mathcal X} \oplus \mathcal B_\infty(0, \epsilon) \right),
  \end{aligned}
  \end{equation*}
  where the last step follows from knowing that $\xd_n(s)$ only takes values $x$ such that $\abstr_\eta(x) = x$. Thus,
  \begin{equation*}
    \sum_{n \in [N]} \hspace{-1mm} \mathds{1}_{\cS_l} \left( \xc_n(s), \u_n(s) \right) 
      \leq \hspace{-2mm} \sum_{n \in [N]} \hspace{-1mm} \mathds{1}_{\mathcal G^{+\epsilon} (\cS_l)} \left( \xd_n(s), \u_n(s) \right)  \hspace{-1mm} \leq \hspace{-1mm} \cN_l,
  \end{equation*}
  which shows that the constraint in \eqref{eq:thm_cont_inst} is satisfied.
\end{proof}

\begin{proof}[Proof of Theorem \ref{thm:cont_to_abs}]

  Suppose for contradiction that there is a solution to \eqref{eq:thm2_cont_inst} but not to \eqref{eq:thm2_graph_inst}. Let $\{ \pi_n \}_{n \in [N]}$ be individual switching policies that solve \eqref{eq:thm2_cont_inst} by generating trajectories $\xc_n(0)\xc_n(1)\ldots$ and actions $\u_n(0)\u_n(1)\ldots$ for $\Sigma_{\tau}$. Due to bisimilarity and the $\eta/2$-proximity of initial conditions, the individual trajectories $\xd_n(0)\xd_n(1)\ldots$ of $\Sigma_{\tau, \eta}$ and the individual trajectories $\xc_n(0) \xc_n(1) \ldots$ of $\Sigma_\tau$ satisfy $\| \xd_n(s) - \xc_n(s) \| \leq \epsilon$ for all $s \in \mathbb{N}$ when the actions $\u_n(0) \u_n(1)\ldots$ are implemented for both systems. For a set $A$ we have $\abstr_\eta \left( A \ominus \mathcal B_\infty\left(0, \frac{\eta}{2} \right) \right) \subset \{ x \in A : \abstr_\eta(x) = x \} \subset A$. Thus, 
  \begin{equation*}
  \begin{aligned}
    & \xd_n(s) \in \abstr_\eta \left( \cS_l^{\mathcal X} \ominus \mathcal B_\infty\left(0, \epsilon + \frac{\eta}{2} \right) \right) \\
    & \implies \xd_n(s) \in  \cS_l^{\mathcal X} \ominus \mathcal B_\infty\left(0, \epsilon \right) \\
    & \implies \xc_n(s) \in \left( \cS_l^{\mathcal X} \ominus \mathcal B_\infty\left(0, \epsilon \right)\right) \oplus \mathcal B_\infty(0, \epsilon)
    \subset \cS_l^{\mathcal X}.
  \end{aligned}
  \end{equation*}
  It follows that
  \begin{equation*}
  \begin{aligned}
      \sum_{n \in [N]} \hspace{-1mm} \mathds{1}_{\mathcal G^{-(\epsilon + \eta/2)} \left( X_l \right)} (\xd_n(s), \u_n(s))
      \leq \hspace{-1mm} \sum_{n \in [N]} \hspace{-1mm} \mathds{1}_{X_l} \left( \xc_n(s), \u_n(s) \right),
  \end{aligned}
  \end{equation*}
  so $\{ \pi_n \}_{n \in [N]}$ is a solution also for \eqref{eq:thm2_graph_inst}---a contradiction.
\end{proof}

\begin{proof}[Proof of Theorem \ref{thm:feas_to_sol}]

	We first consider the case $s \leq T$, and claim that the selection on Line 2 is possible if 
	\begin{equation}
	\label{eq:aggregate}
		w_\ds(s) = \sum_{n \in [N]} \mathds{1}_{ \{ \ds\} } \left( \xd_n(s) \right), \quad \forall \ds \in \mathcal Q.
	\end{equation}
	and, furthermore, that the selection guarantees that \eqref{eq:aggregate} holds at time $(s+1)$. 

	Due to \eqref{eq:aggregate_initial}, equation \eqref{eq:aggregate} holds at $s = 0$. For induction, assume that \eqref{eq:aggregate} holds at time $s$. Then by \eqref{eq:lp_c5},
	\begin{equation}
	\label{eq:total_agrees}
	 	\sum_{\du \in \mathcal U} r_\ds^\du(s) = w_\ds(s) = \sum_{n \in [N]} \mathds{1}_{\{ \ds \}} \left( \xd_n(s) \right).
	 \end{equation}
	The selection on line 2 amounts to for each $\ds \in \mathcal Q$ assigning $w_\ds(s)$ objects to $|\mathcal U|$ ``bins'' such that the $\du$-bin has $r_\ds^\du$ members; by \eqref{eq:total_agrees} this is doable. Remark that if $\u_n(s) = \du$, then $\xd_n(s+1) = \ds$ if and only if $\xd_n(s) \in \mathcal N_\ds^\du$. Thus,
	\begin{equation*}
	\begin{aligned}
			\sum_{n \in [N]} & \mathds{1}_{\{\ds\} } \left( \xd_n(s+1) \right) 
			= \sum_{\substack{n \in [N] \\ \du \in \mathcal U}} \mathds{1}_{\{ (\ds, \du) \} } \left( \xd_n(s+1), \u_n(s) \right) \\
			& = \sum_{n \in [N]} \sum_{\du \in \mathcal U} \sum_{\ds' \in \mathcal N_\ds^\du} \mathds{1}_{\{ (\ds', \du) \} } \left( \xd_n(s), \u_n(s) \right) \\
			& = \sum_{\du \in \mathcal U} \sum_{\ds' \in \mathcal N_\ds^\du} r_{\ds'}^\du(s) = w_\ds(s+1),
	\end{aligned}
	\end{equation*}
	where the last step follows from \eqref{eq:aggregate_dynamics}. Thus \eqref{eq:aggregate} holds for all $s \in [T+1]$. 

	Secondly, we consider the case $s \geq T$ and claim that the selection on line 4 is possible if for all $\ds \in \mathcal Q$
	\begin{equation}
	\label{eq:cycle_agreement}
	\begin{aligned}
		\sum_{\du \in \mathcal U} \sum_{j \in [J]} \left\langle \cycleC_j, \alpha_j^{\circlearrowleft (s-T)} \right\rangle^{ \{(\ds, \du)\}} = \sum_{n \in [N]} \mathds{1}_{\{\ds \} } \left( \xd_n(s)\right),
	\end{aligned}
	\end{equation}
	and, furthermore, that the selection guarantees that \eqref{eq:cycle_agreement} holds at time $(s+1)$. To show that \eqref{eq:cycle_agreement} enables a selection $\{ \u_n(s) \}_{n \in [N]}$ satisfying line 4, it suffices to remark that the selection problem is equivalent to above with $r_\ds^\du(s)$ replaced by $\sum_{j \in [J]} \left\langle \cycleC_j, \alpha_j^{\circlearrowleft (s-T)} \right\rangle^{ \{(\ds, \du)\}}$.

	Due to \eqref{eq:lp_c4} and \eqref{eq:aggregate}, equation \eqref{eq:cycle_agreement} holds at $s = T$. Suppose for induction that \eqref{eq:cycle_agreement} holds at time $s \geq T$ and that a selection $\{ \u_n(s) \}_{n \in [N]}$ satisfying line 4 has been made. Then,
	\begin{equation*}
	\begin{aligned}
			& \sum_{n \in [N]} \mathds{1}_{\{\ds \} } \left( \xd_n(s+1)\right) = \sum_{\substack{n \in [N] \\ \du \in \mathcal U}} \mathds{1}_{\{(\ds, \du) \} } \left( \xd_n(s+1), \u_n(s)\right) \\
			& = \sum_{n \in [N]} \sum_{\du \in \mathcal U} \sum_{\ds' \in \mathcal N_\ds^\du} \mathds{1}_{\{ (\ds', \du) \} } \left( \xd_n(s), \u_n(s) \right) \\
			& = \hspace{-2mm} \sum_{\substack{j \in [J] \\ \du \in \mathcal U}} \hspace{-1mm} \left\langle C_j, \alpha_j^{\circlearrowleft (s - T)} \right\rangle^{ \mathcal N_\ds^\du \times \{\du\} } \hspace{-2mm} = \hspace{-1mm} \sum_{\substack{j \in [J] \\ \du \in \mathcal U}} \hspace{-1mm} \left\langle C_j, \alpha_j^{\circlearrowleft (s + 1 - T)} \right\rangle^{\{ (\ds, \du) \}}.
	\end{aligned}	
	\end{equation*}
	The last step follows from the observation that a node $\ds' \in \mathcal N_\ds^\du$ must have a cycle index $(i - 1) \mod |\cycleC_j|$ in cycle $\cycleC_j$, where $i$ is the cycle index of $\ds$. Thus the selection on line 4 is feasible for all $s \geq T$.

	Finally, we show that each counting constraint $(\cS_l, \cN_l)$ is satisfied. For $s < T$ we have from line 2 and \eqref{eq:lp_c1}:
	\begin{equation*}
	\begin{aligned}
		&\sum_{n \in [N]} \mathds{1}_{\cS_l} \left( \xd_n(s), \u_n(s) \right) \\
		& = \sum_{n \in [N]} \sum_{\ds \in \mathcal Q} \sum_{\du \in \mathcal U} \mathds{1}_{\cS_l} \left( \ds, \du \right) \; \mathds{1}_{\{ (\ds, \du) \}} \left( \xd_n(s), \u_n(s) \right) \\
		& = \sum_{\ds \in \mathcal Q} \sum_{\du \in \mathcal U} \mathds{1}_{\cS_l} \left( \ds, \du \right) \; r_\ds^\du(s) \leq \cN_l.
	\end{aligned}
	\end{equation*}
	Thus the counting constraints are satisfied in the prefix phase. For the suffix phase, from line 4 and \eqref{eq:lp_c2} it follows that
	\begin{equation*}
	\begin{aligned}
		& \sum_{n \in [N]} \sum_{\ds \in \mathcal Q} \sum_{\du \in \mathcal U} \mathds{1}_{\cS_l} \left( \ds, \du \right) \; \mathds{1}_{\{ (\ds, \du) \}} \left( \xd_n(s), \u_n(s) \right) \\
		& = \sum_{\ds \in \mathcal Q} \sum_{\du \in \mathcal U} \mathds{1}_{\cS_l} \left( \ds, \du \right) \sum_{j \in [J]} \left\langle C_j, \alpha_j^{\circlearrowleft (s-T)} \right\rangle^{\{ (\ds, \du) \}} \\
		& = \sum_{j \in [J]} \left\langle C_j, \alpha_j^{\circlearrowleft (s-T)} \right\rangle^{\cS_l} \leq \cN_l.
	\end{aligned}
	\end{equation*}
	Thus the switching protocol in Algorithm \ref{alg:control_extraction} generates inputs and trajectories that satisfy the constraints of \eqref{eq:graph_instance}.
\end{proof}

\begin{proof}[Proof of Theorem \ref{thm:periodic}]

Let $\{ \pi^*_n \}_{n \in [N]}$ be a solution to the instance \eqref{eq:graph_instance} and consider the generated sequence of controls $\u_n(s)$, $s \in \mathbb{N}$, and aggregate states $w_\ds(s) = \sum_{n \in [N]} \mathds{1}_{\{ \ds \}} \left( \xd_n(s) \right)$. The number of possible values $\mathbf{w}(s)$ can take is finite and given by $\binom{|\mathcal Q|+N-1}{N}$---the number of ways in which $N$ identical objects (subsystems) can be partitioned into $|\mathcal Q|$ sets (nodes). We can therefore find times $T_1, T_2$ with $T_1 < T_2 \leq \binom{|\mathcal Q|+N-1}{N}$ such that $\mathbf{w}(T_1) = \mathbf{w}(T_2)$. We show that the graph flows induced by $\{ \pi^*_n \}_{n \in [N]}$ on the time interval $[T_1, T_2]$ can be achieved with cycle assignments: we define a flow on a graph in a higher dimension, decompose it into flows over cycles, and project the cyclic flow onto the original graph.

Let $G = (\mathcal Q, \longrightarrow)$ be the system graph, and define a new graph $H = (V_H, E_H)$. The node set $V_H = \underbrace{\mathcal Q \times \mathcal Q \times \ldots \times \mathcal Q}_{T_2 - T_1 \text{ times}}$ contains $T_2 - T_1$ copies of each node in $\mathcal Q$, and copies of $\ds \in \mathcal Q$ are labeled $\ds_s$ for $s \in [T_1, T_2]$. The set of edges is defined as
\begin{equation*}
\begin{aligned}
		E_H = \left\{ (\ds_s, \tilde \ds_{s+1}) : s \in \{ T_1, \ldots, T_2-1\}, \; (\ds,\tilde \ds) \in E \right\} \\
		 \bigcup \left\{ (\ds_{T_2}, \ds_{T_1}) : \ds \in \mathcal Q \right\}.
\end{aligned}	
\end{equation*}
An edge flow is induced on $H$ by $\{ \pi^*_n \}_{n \in [N]}$, obtained by letting the flow along $(\ds_s, \tilde \ds _{s+1})$ be the number of subsystems that traverses the edge $(\ds,\tilde \ds) \in E$ at time $s$, and by letting the flow along $(\ds_{T_2}, \ds_{T_1})$ be equal to the number of systems at $\ds$ at time $T_1$. By construction, this flow is balanced at each node (i.e. inflows equal outflows).

By the flow decomposition theorem \cite[Theorem 3.5]{Ahuja:1993uh}, we can then find cycles and assignments in $H$ that achieve this edge flow. As becomes evident from the proof in \cite{Ahuja:1993uh}, these cycles are furthermore simple in $H$ and thus of length at most $|V_H| = |\mathcal Q|(T_2-T_1)$. By projecting these cycles onto a single copy of $\mathcal Q$, we obtain cycles and assignments in $G$ that mimic the counting performance of $\{ \pi^*_n \}_{n \in [N]}$ on the interval $[T_1, T_2]$ when circulated. 

We can therefore define a prefix-suffix strategy by taking as the prefix part $r_\ds^\du(s) = \sum_{n \in [N]} \mathds{1}_{\{(\ds, \du) \}} (\xd_n(s), \u_n(s) ), \quad s \in [T_1]$, followed by a suffix part consisting of the cycles and assignments constructed as above.
\end{proof}

\begin{lemma}
\label{lemma:simple}
	Suppose $\cycleC = \cycleC_1 \cup \cycleC_2$ is a cycle that visits a node $\ds_{0}$ twice, so that it can be decomposed into two cycles $\cycleC_1 = (\ds_{0}, \ds_{1}) \ldots (\ds_{{i}}, \ds_{0})$ and $\cycleC_2 = (\ds_{0}, \ds_{{i+1}}) \ldots (\ds_{{|\cycleC|-1}}, \ds_{0})$. 

	Let $\assgn$ be an assignment to $\cycleC$ that satisfies $\maxcnt^{\cS}(\cycleC, \assgn) \leq \cN$, let $P$ be the graph period, and let 
	\begin{equation*}
	\begin{aligned}
		 	& N_p = \sum_{ i \in [|C|/P] } \alpha(p + iP), \quad p \in [P], \\
		 	& N_p^1 =\frac{|\cycleC_1|}{|\cycleC|} N_p, \quad N_p^2 =\frac{|\cycleC_2|}{|\cycleC|} N_p, \quad p \in [P].
	\end{aligned}
	\end{equation*}
	Then the joint $\cS$-counts for the assignment $\bar \alpha_{\{ N_p^1 \}_{p \in [P]}}$ to $\cycleC_1$ and the assignment $\bar \alpha_{\{ N_p^2 \}_{p \in [P]}}$ to $\cycleC_2$, satisfy
	\begin{equation*}
		\maxcnt^{\cS} \left( \{ \cycleC_1, \cycleC_2 \}, \{ \bar \alpha_{\{ N_p^1 \}_{p \in [P]}}, \bar \alpha_{\{ N_p^2 \}_{p \in [P]}} \} \right) \leq \cN.
	\end{equation*}
\end{lemma}

\begin{proof}[Proof of Lemma \ref{lemma:simple}]
	We have for all $i_1$ that
	\begin{equation*}
	\begin{aligned}
		\bar \alpha_{\{ N_p^1 \}_{p \in [P]} }(i_1) = \frac{P N^1_{(i_1 \mod P)}}{|C_1|} = \frac{P N_{(i_1 \mod P)}}{|C|} \\
		= \frac{P}{|C|} \sum_{k \in [|C|/P]} \alpha((i_1 \mod P) + kP),
	\end{aligned}
	\end{equation*}
	and similarly for $C_2$. $P$ must divide both $|C_1|$ and $|C_2|$, so $((i_1 \mod |C_1|) \mod P) = (i_1 \mod P)$ for all $i_1$ and similarly for $|C_2|$. Below, for a cycle $C = (q_0, \mu_0, q_1) (q_1, \mu_1, q_2) \ldots$ we use the short hand notation $\mathds{1}^C_X(i) = 1$ if $(q_i, \mu_i) \in X$ and $\mathds{1}^C_X(i) = 0$ otherwise,
	to indicate that the $i$'th state-action pair is in the counting set $X$. For $j=1,2$ we get
	\begin{equation*}
	\begin{aligned}
		& \left\langle C_j, \bar \alpha_{\{ N_p^j \}_{p \in [P]}}^{\circlearrowleft s} \right\rangle^X = \sum_{i_j \in [|C_j|]} \mathds{1}^{C_j}_X \left( i_j \right) (\bar \alpha_{\{ N_p^j \}_{p \in [P]} })^{\circlearrowleft s}(i_j) \\
		& = \frac{P}{|C|} \sum_{i_j \in [|C_j|]} \mathds{1}^{C_j}_X \left( i_j \right) \hspace{-2mm} \sum_{k \in [|C|/P]} \hspace{-3mm} \alpha(((i_j-s) \mod P) + kP).
	\end{aligned}
	\end{equation*}
	We have $i_2 \mod P = (|C_1| + i_2) \mod P$, and $i_2 \mapsto |C_1| + i_2$ is a mapping from $i_2 \in [|C_2|]$ to the corresponding index in $C$. We can therefore convert a sum over both $i_1$ and $i_2$ into a sum over the index $i$ of $C$.
	\begin{equation*}
	\begin{aligned}
		& \left\langle C_1, \bar \alpha_{\{ N_p^1 \}_{p \in [P]}}^{\circlearrowleft s} \right\rangle^X + \left\langle C_2, \bar \alpha_{\{ N_p^2 \}_{p \in [P]}}^{\circlearrowleft s} \right\rangle^X \\
		& = \frac{P}{|C|} \sum_{i \in [|C|]} \mathds{1}^{C}_X \left( i \right) \sum_{k \in [|C|/P]} \alpha \left( ((i-s) \mod P) + kP \right) \\
		& \leq \max_{k \in [|C|/P]} \sum_{i \in [|C|]} \mathds{1}^C_X \left( i \right) \alpha \left( ((i-s) \mod P) + kP \right) \\
		& \leq \max_{s \in [|C|]} \sum_{i \in [|C|]} \mathds{1}^C_X \left( i \right) \alpha^{\circlearrowleft s} (i) = \maxcnt^\cS(C, \alpha) \leq R.
	\end{aligned}
	\vspace{-9mm}
	\end{equation*}
\end{proof}

\begin{proof}[Proof of Theorem \ref{thm:sufficient}]
	By Theorem \ref{thm:periodic}, we know that a correct solution must eventually lead to periodic behavior, and Lemma \ref{lemma:simple} shows that the suffix of such a solution can be mapped into a suffix on simple cycles consisting of $P$-averaged assignments, where $P$ is the period of the graph\footnote{In the case of a non-connected graph the analysis can be done separately for each strongly connected component.}. What remains to show is that this latter suffix is reachable from the initial state while respecting the relaxed counting constraints.

	Let $\mathbf{w}(s)$ for $s \in \mathbb{N}$ be the aggregate states for an integer solution to \eqref{eq:graph_instance}. From Lemma \ref{lemma:simple} we can obtain a set of simple cycles $I$ such that some $P$-averaged assignments to these cycles satisfy the counting bounds. In addition, these assignments have the same parity structure as $\mathbf{w}(s)$, and hence as the initial condition; therefore they are reachable from the initial condition by virtue of Corollary \ref{cor:periodicity}. We now propose a switching protocol to control the aggregate state to these assignments; the protocol consists in gradually steering a fraction of the systems from the original solution $\mathbf{w}(s)$ to the $P$-averaged assignments pertaining to the cycles in $I$.

	We use the following notation: let $\alpha(s)$ be what remains of the correct trajectory $\mathbf{w}(s)$, let $\beta(s)$ be the fraction currently being steered towards assignments to the simple cycles, and let $\gamma(s)$ be the fraction that has already reached these assignments. We then have $\alpha(0) = \mathbf{w}(0)$, and $\beta(0) = \gamma(0) = 0$. The overall system state is $\tilde{\mathbf{w}}(s) = \alpha(s) + \beta(s) + \gamma(s)$.

	The protocol at time $s+1$ is as follows.
	\begin{itemize}
		\item If $\beta(s) = 0$, set $\alpha(s+1) = \alpha(s) (1 - \epsilon/\| \alpha(s) \|_1) $, and $\beta(s+1) = \alpha(s) \epsilon/\| \alpha(s) \|_1$,
		\item If $\beta(s)$ has reached the average assignments to cycles in $I$, set $\beta(s+1) = 0$, and $\gamma(s+1) = \beta(s) + \gamma(s)$,
		\item Otherwise, steer $\beta(s)$ toward the average assignments to cycles in $I$ .
	\end{itemize}

	We remark that the transitions are all properly connected and merely illustrate the transport of subsystem ``weight'' from $\alpha(s)$ via $\beta(s)$ to the average assignments represented by $\gamma(s)$. 
	Transporting a mass $\epsilon$ takes at most time $(\diam(G)^2+1)$ by \cite{SHEN199621} which gives a bound $T \leq (\diam(G)^2+1)$ for the $T$ in the proof of Theorem \ref{thm:aperiodic}. We can thus infer that the total transportation time is upper bounded by $(\diam(G)^2+1)N/\epsilon$, since an absolute weight $\epsilon$ is transported in each step.

	We finally consider the counting bounds. By assumption, they are satisfied by $\mathbf{w}(s)$, and by Lemma \ref{lemma:simple} they are also satisfied once the average assignments to cycles in $I$ are reached. In the meantime, $\tilde w_\ds^\du(s) = \alpha_\ds^\du(s) + \beta_\ds^\du(s) + \gamma_\ds^\du(s)$. For every $s$, there is an integer $z \leq 1/\epsilon$ such that $\alpha(s) = \left( 1- \frac{\epsilon z}{N} \right) \mathbf{w}(s)$, $\gamma(s) = \frac{\epsilon (z-1)}{N} \gamma_0(s)$,	where $\gamma_0$ is the average assignment to the cycles. Furthermore $\| \beta_\ds^\du(s) \|_1 \leq \epsilon$ which shows that the counts are bounded as follows:
	\begin{equation*}
	\begin{aligned}
			&\sum_{\ds \in \mathcal Q} \sum_{\du \in \mathcal U} \mathds{1}_{{\cS}_l} (\ds, \du) \; \tilde w_\ds^\du(s) \\
			& =  \sum_{\ds \in \mathcal Q} \sum_{\du \in \mathcal U} \mathds{1}_{{\cS}_l} (\ds, \du) \left( \begin{aligned} & (1- \epsilon z/N) w_\ds^\du(s) \\ & + \epsilon (z-1) \gamma_0(s)/N  + \beta_\ds^\du(s) \end{aligned} \right) \\
			&\leq (1 - \epsilon z/N) R_l + \epsilon (z-1) R_l/N + \epsilon \leq \cN_l + \epsilon.
	\end{aligned}
	\end{equation*}
	We have therefore constructed a prefix-suffix solution, where the suffix part consists of simple cycles, such that the relaxed counting bounds $({\cS}_l, \cN_l+\epsilon)$ are satisfied.
\end{proof}

\begin{IEEEbiography}[{\includegraphics[width=1in,height=1.25in,clip,keepaspectratio]{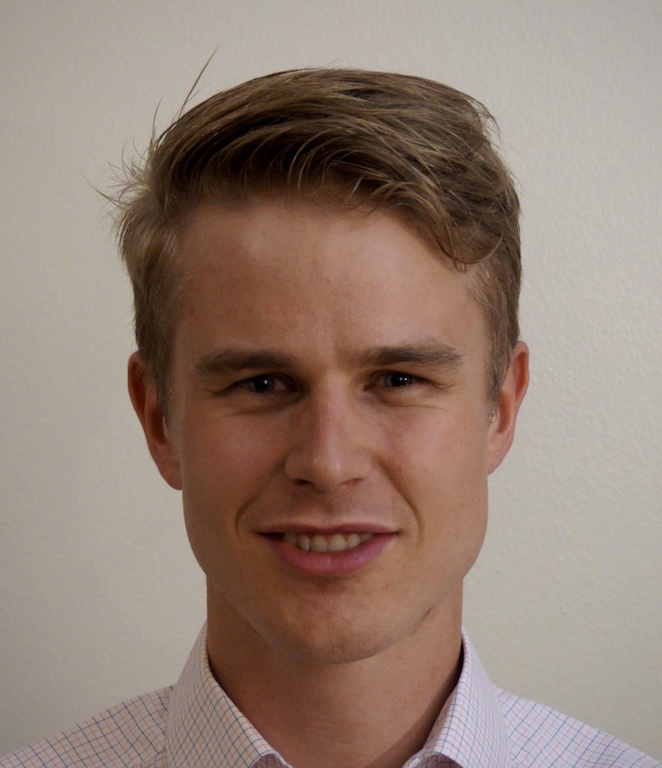}}]{Petter Nilsson} received his B.S. in Engineering Physics in 2011, and his M.S. in Optimization and Systems Theory in 2013, both from KTH Royal Institute of Technology in Stockholm, Sweden, and his Ph.D. in Electrical Engineering in 2017 from the University of Michigan. In addition to his technical degrees, he holds a B.S. in Business and Economics from the Stockholm School of Economics. He is currently a postdoctoral scholar at the California Institute of Technology. 

\end{IEEEbiography}

\begin{IEEEbiography}[{\includegraphics[width=1in,height=1.25in,clip,keepaspectratio]{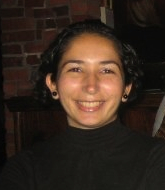}}]{Necmiye Ozay} received the B.S. degree from Bogazici University, Istanbul in 2004, the M.S. degree from the Pennsylvania State University, University Park in 2006 and the Ph.D. degree from Northeastern University, Boston in 2010, all in electrical engineering. She was a postdoctoral scholar at California Institute of Technology, Pasadena between 2010 and 2013. She is currently an assistant professor of Electrical Engineering and Computer Science, at University of Michigan, Ann Arbor. 
\end{IEEEbiography}

\end{document}

%% file: defs.tex
\newcommand{\cycleC}{C}
\newcommand{\cS}{X}
\newcommand{\cN}{R}

\newcommand{\assgn}{\alpha}
\newcommand{\abstr}{\kappa}

\newcommand{\munderbar}{\underline}
\newcommand{\moverbar}{\overline}

\DeclareMathOperator{\diam}{diam}
\DeclareMathOperator{\lcm}{lcm}
\DeclareMathOperator{\maxcnt}{maxcnt}

\newcommand{\off}{\texttt{off}}
\newcommand{\on}{\texttt{on}}

\newcommand{\ds}{q}			
\newcommand{\du}{\mu}		

\newcommand{\xc}{\mathbf{x}}			
\newcommand{\dc}{\mathbf{d}}			
\newcommand{\xd}{\xi}		
\renewcommand{\u}{\sigma}	

%% file: figures/histogram.tex
\begin{tikzpicture}[
    dot/.style = {ellipse, fill, inner sep=0, minimum width=2.55mm, minimum height=0.75mm, rotate=10}]
    \fill[fill=gray!20] (-0.5,0,-0.5) -- (6.5,0,-0.5) -- (6.5,0,3.5) -- (-0.5,0,3.5) -- cycle;
    \draw[dashed] (0,0,0) -- (0,0,3);
    \draw[dashed] (1.5,0,0) -- (1.5,0,3);
    \draw[dashed] (3,0,0) -- (3,0,3);
    \draw[dashed] (4.5,0,0) -- (4.5,0,3);
    \draw[dashed] (6,0,0) -- (6,0,3);
    \draw[dashed] (0,0,0) -- (6,0,0);
    \draw[dashed] (0,0,1.5) -- (6,0,1.5);
    \draw[dashed] (0,0,3) -- (6,0,3);

    \draw[latex-] (6.25, 0, -0) to[out=90, in=-90] ++ (-0.5,2,0) node[above] {$\mathcal X$};

    \node at (1.25,0,2.6) {$q_0$};
    \node at (2.75,0,2.6) {$q_1$};
    \node at (4.25,0,2.6) {$q_2$};
    \node at (5.75,0,2.6) {$q_3$};

    \node at (1.25,0,0.4) {$q_4$};
    \node at (2.75,0,0.4) {$q_5$};
    \node at (4.25,0,0.4) {$q_6$};
    \node at (5.75,0,0.4) {$q_7$};

    \node[dot] at (2.4108, 0, 1.8256) {};
    \node[dot] at (0.4558, 0, 0.9793) {};
    \node[dot] at (1.0558, 0, 0.7793) {};
    \node[dot] at (2.4395, 0, 2.7043) {};
    \node[dot] at (0.7399, 0, 1.0595) {};
    \node[dot] at (0.9034, 0, 2.4636) {};
    \node[dot] at (1.8397, 0, 1.262) {};
    \node[dot] at (2.5036, 0, 0.7291) {};
    \node[dot] at (0.2979, 0, 0.5070) {};
    \node[dot] at (5.4163, 0, 1.9473) {};
    \node[dot] at (5.6687, 0, 2.1952) {};
    \node[dot] at (2.7452, 0, 1.9432) {};
    \node[dot] at (3.4355, 0, 1.3528) {};


    \footnotesize

    \draw[-latex] (0,0,0) -- (0,2.5,0) node[above] {$w_q$};
    \draw (0.1,0,0) -- (-0.1,0,0) node[left] {0};
    \draw (0.1,0.5,0) -- (-0.1,0.5,0) node[left] {1};
    \draw (0.1,1,0) -- (-0.1,1,0) node[left] {2};
    \draw (0.1,1.5,0) -- (-0.1,1.5,0) node[left] {3};
    \draw (0.1,2,0) -- (-0.1,2,0) node[left] {4};

    \pgfmathsetmacro{\cubex}{1.5}
    \pgfmathsetmacro{\cubez}{1.5}
    \pgfmathsetmacro{\cubey}{2}
    \draw[black,opacity=0.3] (0,0,0) coordinate (o) -- ++(\cubex,0,0) coordinate (a) -- ++(0,\cubey,0) -- ++(-\cubex, 0, 0) -- cycle -- ++ (0,0,\cubez) coordinate[pos=1] (end1) {};
    \draw[black, fill=blue!50, opacity=0.3] (end1) -- ++(\cubex,0,0) -- ++(0,0,-\cubez) -- ++(0,\cubey,0) -- ++ (0,0,\cubez) -- ++(0,-\cubey,0);
    \draw[black, fill=blue!50, opacity=0.3] (end1) -- ++(\cubex,0,0) -- ++(0,\cubey,0) -- ++(-\cubex,0,0) coordinate[pos=1] (end2) -- cycle;
    \draw[black, fill=blue!50, opacity=0.3] (end2) -- ++(\cubex,0,0) -- ++(0,0,-\cubez) -- ++(-\cubex,0,0) -- cycle;

    \pgfmathsetmacro{\cubey}{0.5}
    \draw[black, opacity=0.3] (0,0,1.5) coordinate (o) -- ++(\cubex,0,0) coordinate (a) -- ++(0,\cubey,0) -- ++(-\cubex, 0, 0) -- cycle -- ++ (0,0,\cubez) coordinate[pos=1] (end1) {};
    \draw[black, fill=blue!50, opacity=0.3] (end1) -- ++(\cubex,0,0) -- ++(0,0,-\cubez) -- ++(0,\cubey,0) -- ++ (0,0,\cubez) -- ++(0,-\cubey,0);
    \draw[black, fill=blue!50, opacity=0.3] (end1) -- ++(\cubex,0,0) -- ++(0,\cubey,0) -- ++(-\cubex,0,0) coordinate[pos=1] (end2) -- cycle;
    \draw[black, fill=blue!50, opacity=0.3] (end2) -- ++(\cubex,0,0) -- ++(0,0,-\cubez) -- ++(-\cubex,0,0) -- cycle;

    \pgfmathsetmacro{\cubey}{1}
    \draw[black, opacity=0.3] (1.5,0,0) coordinate (o) -- ++(\cubex,0,0) coordinate (a) -- ++(0,\cubey,0) -- ++(-\cubex, 0, 0) -- cycle -- ++ (0,0,\cubez) coordinate[pos=1] (end1) {};
    \draw[black, fill=blue!50, opacity=0.3] (end1) -- ++(\cubex,0,0) -- ++(0,0,-\cubez) -- ++(0,\cubey,0) -- ++ (0,0,\cubez) -- ++(0,-\cubey,0);
    \draw[black, fill=blue!50, opacity=0.3] (end1) -- ++(\cubex,0,0) -- ++(0,\cubey,0) -- ++(-\cubex,0,0) coordinate[pos=1] (end2) -- cycle;
    \draw[black, fill=blue!50, opacity=0.3] (end2) -- ++(\cubex,0,0) -- ++(0,0,-\cubez) -- ++(-\cubex,0,0) -- cycle;

    \pgfmathsetmacro{\cubey}{1.5}
    \draw[black, opacity=0.3] (1.5,0,1.5) coordinate (o) -- ++(\cubex,0,0) coordinate (a) -- ++(0,\cubey,0) -- ++(-\cubex, 0, 0) -- cycle -- ++ (0,0,\cubez) coordinate[pos=1] (end1) {};
    \draw[black, fill=blue!50, opacity=0.3] (end1) -- ++(\cubex,0,0) -- ++(0,0,-\cubez) -- ++(0,\cubey,0) -- ++ (0,0,\cubez) -- ++(0,-\cubey,0);
    \draw[black, fill=blue!50, opacity=0.3] (end1) -- ++(\cubex,0,0) -- ++(0,\cubey,0) -- ++(-\cubex,0,0) coordinate[pos=1] (end2) -- cycle;
    \draw[black, fill=blue!50, opacity=0.3] (end2) -- ++(\cubex,0,0) -- ++(0,0,-\cubez) -- ++(-\cubex,0,0) -- cycle;

    \pgfmathsetmacro{\cubey}{0.5}
    \draw[black, opacity=0.3] (3,0,0) coordinate (o) -- ++(\cubex,0,0) coordinate (a) -- ++(0,\cubey,0) -- ++(-\cubex, 0, 0) -- cycle -- ++ (0,0,\cubez) coordinate[pos=1] (end1) {};
    \draw[black, fill=blue!50, opacity=0.3] (end1) -- ++(\cubex,0,0) -- ++(0,0,-\cubez) -- ++(0,\cubey,0) -- ++ (0,0,\cubez) -- ++(0,-\cubey,0);
    \draw[black, fill=blue!50, opacity=0.3] (end1) -- ++(\cubex,0,0) -- ++(0,\cubey,0) -- ++(-\cubex,0,0) coordinate[pos=1] (end2) -- cycle;
    \draw[black, fill=blue!50, opacity=0.3] (end2) -- ++(\cubex,0,0) -- ++(0,0,-\cubez) -- ++(-\cubex,0,0) -- cycle;

    \pgfmathsetmacro{\cubey}{1}
    \draw[black, opacity=0.3] (4.5,0,1.5) coordinate (o) -- ++(\cubex,0,0) coordinate (a) -- ++(0,\cubey,0) -- ++(-\cubex, 0, 0) -- cycle -- ++ (0,0,\cubez) coordinate[pos=1] (end1) {};
    \draw[black, fill=blue!50, opacity=0.3] (end1) -- ++(\cubex,0,0) -- ++(0,0,-\cubez) -- ++(0,\cubey,0) -- ++ (0,0,\cubez) -- ++(0,-\cubey,0);
    \draw[black, fill=blue!50, opacity=0.3] (end1) -- ++(\cubex,0,0) -- ++(0,\cubey,0) -- ++(-\cubex,0,0) coordinate[pos=1] (end2) -- cycle;
    \draw[black, fill=blue!50, opacity=0.3] (end2) -- ++(\cubex,0,0) -- ++(0,0,-\cubez) -- ++(-\cubex,0,0) -- cycle;


\end{tikzpicture}

%% file: figures/assignment_circ0.tex
\begin{tikzpicture}[
  node distance=1.5cm
  ]
  \node (empty) [circle, minimum width=7mm, draw] at (0,1) {};
  \node (4) [circle, minimum width=7mm, draw, right of=empty, label=$\ds_4$] {$2$};
  \node (3) [circle, minimum width=7mm, draw, right of=4, label=$\ds_3$] {$3$};
  \node (0) [circle, minimum width=7mm, draw, below of=empty, label=below:$\ds_0$] {$6$};
  \node (1) [circle, minimum width=7mm, draw, right of=0, label=below:$\ds_1$] {$5$};
  \node (2) [circle, minimum width=7mm, draw, right of=1, label=below:$\ds_2$] {$4$};
  \draw[-latex] (4)  -- node[right] {$\mu_4$} (0);
  \draw[-latex] (3)  -- node[above, pos=0.25] {$\mu_3$} (4);
  \draw[-latex] (0) -- node[below] {$\mu_0$} (1);
  \draw[-latex] (1) -- node[below] {$\mu_1$} (2);
  \draw[-latex] (2) -- node[left]  {$\mu_2$} (3);

  \node at (2.2,-1.2) {$X$};

  \begin{scope}[on background layer]
    \draw[fill=green!20, dashed] ($(3.north east) + (0.2, 0.5)$) -- ($(3.north west) + (-0.2, 0.5)$) 
    -- ($(1.north west) + (-0.15, 0.1)$) -- ($(1.south west) + (-0.15, -0.6)$)
    -- ($(2.south east) + (0.2, -0.6)$) -- cycle;
  \end{scope}

\end{tikzpicture}

%% file: figures/assignment_circ1.tex
\begin{tikzpicture}[
  node distance=1.5cm
  ]
  \node (empty) [circle, minimum width=7mm, draw] at (0,1) {};
  \node (4) [circle, minimum width=7mm, draw, right of=empty, label=$\ds_4$] {$3$};
  \node (3) [circle, minimum width=7mm, draw, right of=4, label=$\ds_3$] {$4$};
  \node (0) [circle, minimum width=7mm, draw, below of=empty, label=below:$\ds_0$] {$2$};
  \node (1) [circle, minimum width=7mm, draw, right of=0, label=below:$\ds_1$] {$6$};
  \node (2) [circle, minimum width=7mm, draw, right of=1, label=below:$\ds_2$] {$5$};
  \draw[-latex] (4)  -- node[right] {$\mu_4$} (0);
  \draw[-latex] (3)  -- node[above, pos=0.25] {$\mu_3$} (4);
  \draw[-latex] (0) -- node[below] {$\mu_0$} (1);
  \draw[-latex] (1) -- node[below] {$\mu_1$} (2);
  \draw[-latex] (2) -- node[left]  {$\mu_2$} (3);

  \begin{scope}[on background layer]
    \draw[fill=green!20, dashed] ($(3.north east) + (0.2, 0.5)$) -- ($(3.north west) + (-0.2, 0.5)$) 
    -- ($(1.north west) + (-0.15, 0.1)$) -- ($(1.south west) + (-0.15, -0.6)$)
    -- ($(2.south east) + (0.2, -0.6)$) -- cycle;
  \end{scope}

  \node at (2.2,-1.2) {$X$};

\end{tikzpicture}

%% file: figures/averaging.tex
\begin{tikzpicture}[
      every node/.style = {node distance=2.5em, minimum width=2.2em, minimum height=1.5em}
    ]
    \node [minimum width=4em] (I) {$i$};

    \node [right of=I, node distance=3em] (i11) {$0$};
    \node [right of=i11] (i12) {$1$};
    \node [right of=i12] (i13) {$2$};
    \node [right of=i13] (i14) {$3$};
    \node [right of=i14] (i15) {$4$};
    \node [right of=i15] (i16) {$5$};

    \node [below of = I, minimum width=4em, node distance = 1.8em] (a1) {$\alpha$};

    \node [right of=a1, node distance=3em] (a11) {$2$};
    \node [right of=a11] (a12) {$0$};
    \node [right of=a12] (a13) {$1$};
    \node [right of=a13] (a14) {$4$};
    \node [right of=a14] (a15) {$3$};
    \node [right of=a15] (a16) {$2$};

    \node [below of=a1, minimum width=4em, node distance = 1.2em] (a2) {$\bar \alpha_{12}$};

    \node [right of=a2, node distance=3em]  (a21) {$2$};
    \node [right of=a21] (a22) {$2$};
    \node [right of=a22] (a23) {$2$};
    \node [right of=a23] (a24) {$2$};
    \node [right of=a24] (a25) {$2$};
    \node [right of=a25] (a26) {$2$};

    \node [below of=a2, minimum width=4em, node distance = 1.2em] (a3) {$\bar \alpha_{\{3,9\}}$};

    \node [right of=a3, node distance=3em]  (a31) {$1$};
    \node [right of=a31] (a32) {$3$};
    \node [right of=a32] (a33) {$1$};
    \node [right of=a33] (a34) {$3$};
    \node [right of=a34] (a35) {$1$};
    \node [right of=a35] (a36) {$3$};

    \node [below of=a3, minimum width=4em, node distance = 1.2em] (a4) {$\bar \alpha_{\{2,4,6\} }$};

    \node [right of=a4, node distance=3em]  (a41) {$1$};
    \node [right of=a41] (a42) {$2$};
    \node [right of=a42] (a43) {$3$};
    \node [right of=a43] (a44) {$1$};
    \node [right of=a44] (a45) {$2$};
    \node [right of=a45] (a46) {$3$};

    \draw (I.north east) -- (a4.south east);
    \draw (I.south west) -- (i16.south east);

\end{tikzpicture}

%% file: figures/rounding_left.tex
\begin{tikzpicture}
  \begin{axis}[%
    axis x line=center,
    xmin=-0.5,xmax=8.5,
    ymin=-1,ymax=1,
    xtick = {0,1,2,...,8},
    xticklabels={0,1,2,3,4,5,6,0,1},
    axis y line=none,
    tickwidth = 10pt
    ]
    \node[draw,circle, inner sep=0, minimum width=6pt] at (axis cs: 0, 0) {};
    \node[draw,circle, inner sep=0, minimum width=6pt] at (axis cs: 2.33, 0) {};
    \node[draw,circle, inner sep=0, minimum width=6pt] at (axis cs: 4.66, 0) {};
    \node[draw,circle, inner sep=0, minimum width=6pt] at (axis cs: 7, 0) {};
    \node[draw,circle, inner sep=0, minimum width=4pt, fill] at (axis cs: 0, 0) {};
    \node[draw,circle, inner sep=0, minimum width=4pt, fill] at (axis cs: 2, 0) {};
    \node[draw,circle, inner sep=0, minimum width=4pt, fill] at (axis cs: 4, 0) {};
    \node[draw,circle, inner sep=0, minimum width=4pt, fill] at (axis cs: 7, 0) {};

    \draw[-latex] (axis cs:2.33,0.05) -- (axis cs: 2,0.05);
    \draw[-latex] (axis cs:4.66,0.05) -- (axis cs: 4,0.05);

    \draw[decoration={brace,raise=5pt},decorate] 
        (axis cs: 0,0.05) -- node[anchor=south, yshift=6pt] {$d$} (axis cs: 2.33,0.05);
    \draw[decoration={brace,raise=5pt},decorate] 
        (axis cs: 2.33,0.05) -- node[anchor=south, yshift=6pt] {$d$} (axis cs: 4.66,0.05);
    \draw[decoration={brace,raise=5pt},decorate] 
        (axis cs: 4.66,0.05) -- node[anchor=south, yshift=6pt] {$d$} (axis cs: 7,0.05);
  \end{axis}
\end{tikzpicture}

%% file: figures/nl_counting.tex
\begin{tikzpicture}
    \begin{axis}[
      xmin = 0,
      xmax = 50,
      ymin = 0.425,
      ymax = 0.575,
        width=0.8\columnwidth,
        height=0.5\columnwidth,
        xlabel={$t$},
        legend style={
      at={(1.05,1)},
      anchor=north west}
        ]
        \addplot[opacity=0.75, line width=1pt, green] table [x index=0, y index=1, col sep=comma] {data/nl_counts.txt};
        \addplot[opacity=0.75, line width=1pt, yellow] table [x index=0, y index=2, col sep=comma] {data/nl_counts.txt};
        \addplot[opacity=0.75, line width=1pt, red] table [x index=0, y index=3, col sep=comma] {data/nl_counts.txt};

        \legend{$\tilde \cS_1$ count, $\tilde \cS_2$ count, $\cS_{\du_{1}}$ count},
        \draw[dashed] (axis cs: 0, 0.55) -- (axis cs: 55, 0.55);
        \draw[dashed] (axis cs: 0, 0.45) -- (axis cs: 55, 0.45);
    \end{axis}
\end{tikzpicture}

%% file: figures/nl_cycles.tex
\begin{tikzpicture}
    \begin{axis}[
      xmin = -1.2,
      xmax  =1.2,
      ymin = -1,
      ymax = 1,
      xlabel={$x_1$},
      ylabel={$x_2$},
        width=0.9\columnwidth,
        height=0.5\columnwidth
        ]

        \draw[fill=blue!20!white, draw=none] (-1.5,-1) rectangle (0,1);
        \node at (-1,0.8) {$\cS_1$};
        \draw[fill=green!20!white, draw=none] (1.5,-1) rectangle (0,1);
        \node at (1,-0.8) {$\cS_2$};

        \addplot[line width=1pt, red] table [x index=0, y index=1, col sep=comma] {data/cycle155.txt} -- cycle;
        \addplot[line width=1pt, blue] table [x index=0, y index=1, col sep=comma] {data/cycle171.txt} -- cycle;
        \addplot[line width=1pt, green] table [x index=0, y index=1, col sep=comma] {data/cycle177.txt} -- cycle;
        \addplot[line width=1pt, yellow] table [x index=0, y index=1, col sep=comma] {data/cycle195.txt} -- cycle;
        \addplot[line width=1pt, purple] table [x index=0, y index=1, col sep=comma] {data/cycle4.txt} -- cycle;
    \end{axis}
\end{tikzpicture}

%% file: figures/tcl_density1.tex
\begin{tikzpicture}
    \begin{axis}[
        width = \columnwidth,
        height = 0.4\columnwidth,
        enlargelimits=false,
        axis on top,
        ylabel = {Temp. [$^\circ$C]}]
    \addplot graphics [xmin=0,xmax=5,ymin=21,ymax=24] {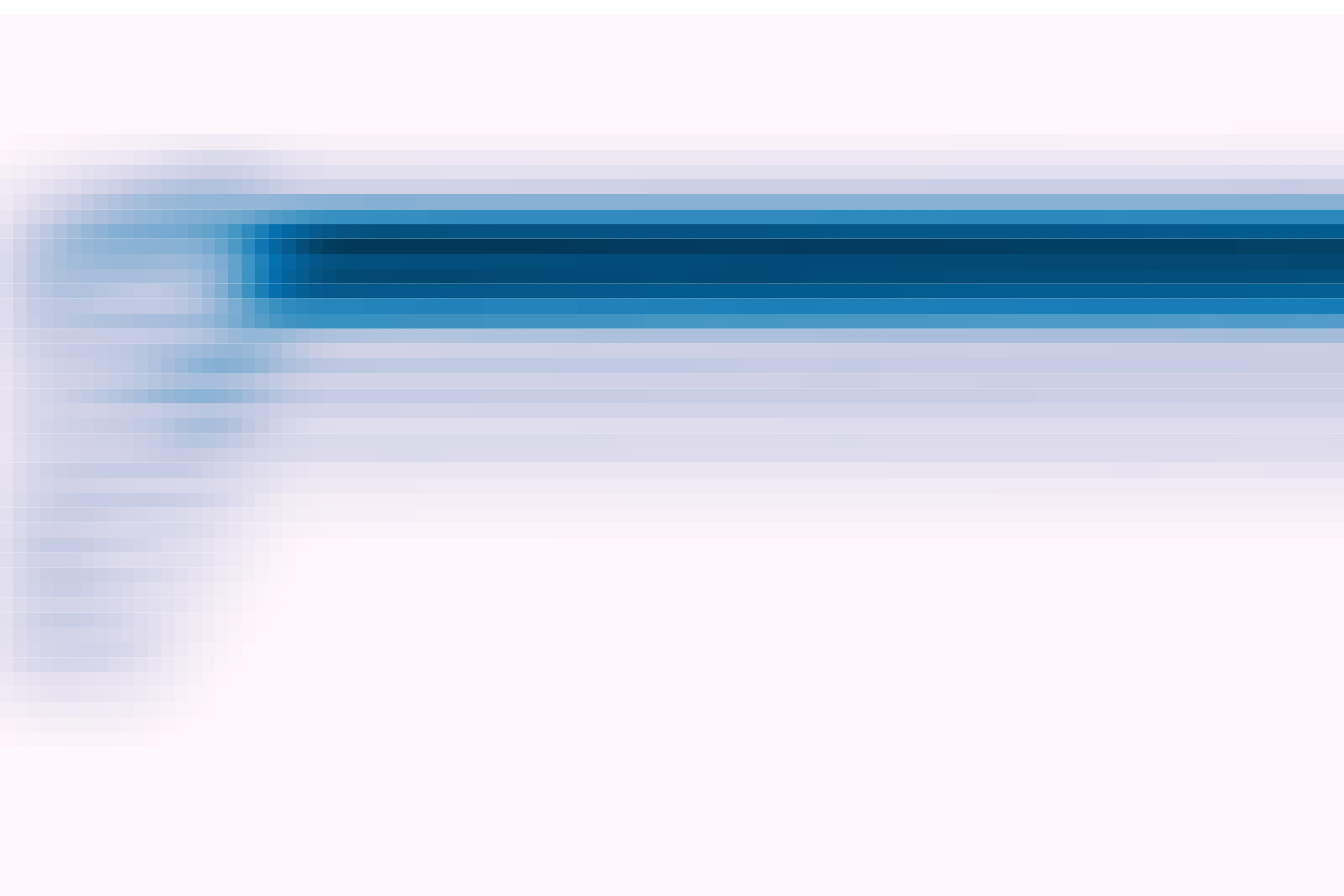};
    \draw[dashed] (axis cs: 0, 23.7) -- (axis cs: 5, 23.7);
    \draw[dashed] (axis cs: 0, 21.3) -- (axis cs: 5, 21.3);
    \end{axis}
\end{tikzpicture}

%% file: figures/tcl_density2.tex
\begin{tikzpicture}
    \begin{axis}[
        width = \columnwidth,
        height = 0.4\columnwidth,
        enlargelimits=false,
        axis on top,
        xlabel = {Time [hrs]},
        ylabel = {Temp. [$^\circ$C]}]
    \addplot graphics [xmin=0,xmax=5,ymin=21,ymax=24] {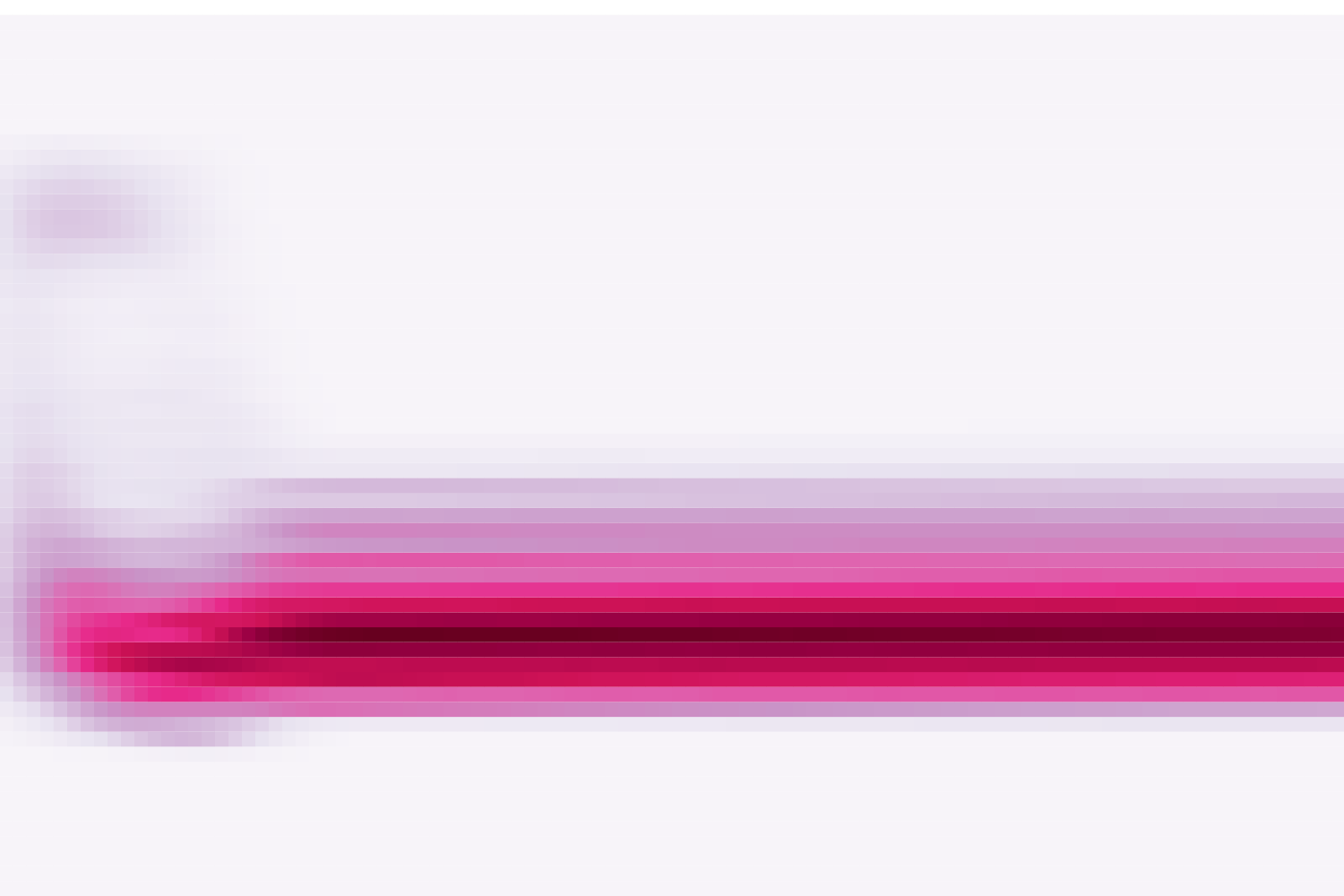};
    \draw[dashed] (axis cs: 0, 23.7) -- (axis cs: 5, 23.7);
    \draw[dashed] (axis cs: 0, 21.3) -- (axis cs: 5, 21.3);
    \end{axis}
\end{tikzpicture}

%% file: figures/tcl_count.tex
\begin{tikzpicture}
    \begin{axis}[
        width=0.9375\columnwidth,
        height=0.4\columnwidth,
        xmin=0,
        xmax=5,
        ymin=1000,
        ymax=12000,
        y tick label style={/pgf/number format/.cd,%
          scaled y ticks = false,
          fixed},
        xlabel={Time [hrs]},
        ylabel={mode-$\texttt{on}$-count}]
        \addplot[const plot, blue, line width=1pt] plot file {data/count_low.txt};
        \addplot[const plot, red, line width=1pt] plot file {data/count_high.txt};
        \draw[dashed] (axis cs: 0, 6000) -- (axis cs: 5, 6000);
        \draw[dashed] (axis cs: 0, 6700) -- (axis cs: 5, 6700);
    \end{axis}
\end{tikzpicture}